\documentclass{siamart190516}
\usepackage[utf8]{inputenc}
\usepackage{graphicx}
\usepackage{xspace}
\usepackage{todonotes}
\usepackage{mathtools, amsfonts,mathrsfs}
\usepackage{cite}
\usepackage{microtype}
\usepackage{caption, subcaption}
\usepackage{algorithm}
\usepackage[]{algpseudocode}
\usepackage{textcomp, mathcomp}
\usepackage{array}
\usepackage{comment}
\usepackage{amsmath,amssymb}

\usepackage{hyperref}
\usepackage{bibentry}
\usepackage{bbm}

\graphicspath{{./fig/}}
\newcommand{\R}{\mathbb{R}}
\newcommand{\N}{\mathbb{N}}
\newcommand{\Q}{\mathcal{Q}}
\newcommand{\eps}{\varepsilon}
\newcommand{\OPT}{\textnormal{OPT}}

\newcommand{\MIN}{\textnormal{MIN}}

\newcommand{\MST}{\textnormal{MST}}
\newcommand{\ST}{\textnormal{ST}}

\newcommand{\COST}{\textnormal{COST}}

\newcommand{\ceiling}[1]{\lceil #1 \rceil}
\newcommand{\C}{\mathscr{C}}
\newcommand{\Schwartz}{\mathscr{S}}
\newcommand{\PP}{\mathcal{P}}
\newcommand{\argmin}{\textnormal{argmin}}

\newsiamthm{problem}{Problem}
\newsiamremark{remark}{Remark}

\begin{document}
\allowdisplaybreaks

\title{Multi-Level Graph Sketches via Single-Level Solvers}

\author{Reyan Ahmed\thanks{Department of Computer Science, University of Arizona, Tucson, AZ 85721
(\email{abureyanahmed@email.arizona.edu, kobourov@cs.arizona.edu, faryad@cs.arizona.edu, rcspence@email.arizona.edu}).}
\and Keaton Hamm\thanks{Department of Mathematics, University of Arizona, Tucson, AZ 85721 (\email{hamm@math.arizona.edu},\email{mjlatifi@math.arizona.edu}.)}
\and Stephen Kobourov\footnotemark[1]
\and Mohammad Javad Latifi Jebelli\footnotemark[2]
\and Faryad Darabi Sahneh\footnotemark[1]
\and Richard Spence\footnotemark[1]
}

\maketitle
\begin{abstract}
 Given an undirected weighted graph $G(V,E)$, a constrained sketch over a terminal set $T\subset V$ is a subgraph $G'$ that connects the terminal vertices while satisfying a given set of constraints.  Examples include Steiner trees (preserving connectivity among $T$) and subsetwise spanners (preserving shortest path distances up to a stretch factor).  Multi-level constrained terminal sketches are generalizations in which terminal vertices require different levels or grades of service and each terminal pair is connected with edges of at least the minimum required level of the two terminals.
 
 This paper gives a flexible procedure for computing a broad class of multi-level graph sketches, which encompasses multi-level graph spanners, Steiner trees, and $k$--connected subgraphs as a few special cases. The proposed procedure is modular, i.e., it relies on availability of a single-level solver module (be it an oracle or approximation algorithm). One key result is that an $\ell$--level constrained terminal sketch can be computed with $\log\ell$ queries of the solver module while producing feasible solutions with approximation guarantees independent of $\ell$.
 
 Additionally, a new polynomial time algorithm for computing a subsetwise spanner is proposed.  We show that for $k\in\N$, $\eps>0$, and $T\subset V$, there is a subsetwise $(2k-1)(1+\eps)$--spanner with total weight $O(|T|^\frac1kW(\ST(G,T)))$, where $W(\ST(G,T))$ is the weight of the Steiner tree of $G$ over the subset $T$.  This is the first algorithm and weight guarantee for a multiplicative subsetwise spanner for nonplanar graphs.   Numerical experiments are also done to illustrate the performance of the proposed algorithms.
\end{abstract}

\section{Introduction}

Graph sketching has become a fundamental technique for analysis and visualization of large graphs.  Here, a \textit{sketch} of a large graph on $n$ vertices is a subgraph which is much sparser than the original graph either in terms of the number of edges if the graph is unweighted or the total weight if the graph is weighted. Classic sketches include minimum spanning trees, Steiner trees, and spanners.  These objects have been used as tools in a variety of tasks including network routing~\cite{abraham2018ramsey, roditty2002roundtrip, TZ01, censor2018sparsest,elkin2014light}, computational biology~\cite{BANDELT1986309}, synchronizers~\cite{awerbuch1985complexity}, broadcasting~\cite{Peleg2000Distributed}, wireless sensor networks~\cite{5061918}, robotics control optimization~\cite{DBLP:journals/ijcga/CaiK97}, and computing shortest path distances~\cite{sommer2014shortest}.

\subsection{Constrained Terminal Graph Sketches}

Often, one is prescribed a set of \textit{terminals}, which are important designated vertices of the input graph. For example, in telecommunications networks, terminals may be high-traffic cell towers; in airline flight networks, terminals may be large hub airports for a given airline.  

Given a connected, undirected graph $G(V,E)$ and a set of terminals $T\subset V$, a \textit{terminal sketch of $G$ over $T$} is a connected subgraph $G'(V',E')$ which spans $T$; i.e., $T\subset V'$.  Often in applications, one desires not just a subgraph which spans the terminals, but rather one that satisfies a set of constraints; we will call these \textit{constrained terminal sketches}.  For example, a \textit{Steiner tree} of $G$ over the terminals $T$ has the single constraint that $G'$ must be a tree.  More sophisticated constrained terminal sketches require that the subgraph maintain some structural feature of $G$; for example, a \textit{distance preserver} over $T$ is a subgraph $G'$ spanning $T$ such that pairwise distances in $G'$ between terminals are the same as the pairwise distances in $G$ between terminals.  

For notational purposes, given a graph $G$, let $\PP(V)$ denote the power set of $V$, i.e., the set of all possible subsets of vertices, and via a minor abuse of notation $\PP(G)$ will be the set of all connected subgraphs of $G$. We will assign each edge $e \in E$ a \emph{distance} $d(e) \in \R_+$ as well as a \emph{cost} $c(e)\in\R_+$ ($\R_+$ is the set of positive real numbers), and always assume $G$ is connected and undirected.

\begin{definition}[Constrained Terminal Sketch]\label{DEF:CTerminalSketch}
Given a graph $G(V,E)$, a set of terminals $T\subset V$, a distance function $d: E \to \R_+$, a cost function $c: E \to \R_+$, and a set of constraints $\C$, a subgraph $G'(V',E')\in\PP(G)$ is a \textit{$\C$--terminal sketch of $G$ over $T$} provided $T\subset V'$ and $G'$ satisfies the constraints $\C$ (recall that by definition of $\PP(G)$, $G'$ must also be connected).

Moreover, define
\[\Schwartz(G,T,\C):=\{G'\in\PP(G): G' \textnormal{ is a } \C\textnormal{--terminal sketch of } G \textnormal{ over } T\}. \]           
\end{definition}
The Steiner tree problem is the case where $\C$ requires that $G'$ is a tree. Since $\C$ is an arbitrary constraint set, we should choose $\C$ carefully; if $\C$ is the constraint ``$G'$ has fewer than $|T|-1$ edges'', then no solution exists as $G'$ should also be connected, and $|\Schwartz(G,T,\C)|=0$. Likewise, it is possible for $\C$--terminal sketches of $G$ to exist for some subsets $T$ of terminals but not others (e.g., $\C$ is the constraint ``$G'$ has fewer than 17 edges''). Here, we stipulate that given constraints $\C$, there exists a $\C$--terminal sketch of $G$ over $T$ for all subsets $T \subset V$.

\subsection{Constrained Terminal Graph Sketching Problems}

There are two natural problems to consider for graph sketches. Given an input $(G, T, \C)$, one can seek to solve the decision problem of whether there exists a subgraph $G'$ which spans $T$ and satisfies constraints $\C$.  Likewise, one may attempt to solve the optimization problem of finding the sparsest terminal sketch $G'$.  Most variants of these problems are known to be NP--complete and NP--hard, respectively. Indeed, the Steiner tree problem was one of Karp's original list of 21 NP--complete problems \cite{Karp1972}.  This article primarily focuses on the optimization problem which we state formally here.

\begin{problem}[Minimum Cost Constrained Terminal Sketch problem]\label{PROB:ConstrainedSketch}
Given as input $( G, T, c, d, \C)$, find a minimum-cost connected subgraph $G'(V',E')$ which spans $T$ and satisfies constraints $\C$.  That is, solve 
\[\underset{G'\in\Schwartz(G,T,\C)}\argmin\; \COST(G')\; := \underset{G'\in\Schwartz(G,T,\C)}\argmin\; \sum_{e\in E'}c(e).\]
\end{problem}

\Cref{PROB:ConstrainedSketch} is very generic as all of the aforementioned terminal sketches can be formulated as an instance of this problem. Here, we will use both cost and distance in the multi-level formulation, so we prefer cost to be that part of the objective function which is minimized.  Note that some problems (like computing Steiner trees) do not use the distance, whereas others do (like spanners discussed momentarily).

\subsection{Multi-level Constrained Terminal Sketches}

Multi-level graph representations have long been used for visualization of complex networks. For example, map applications implicitly use such a representation, as the user is allowed to zoom in which adds vertices (intersections) and edges (roads) at a finer scale or to zoom out which takes out less important connections.  Take as an example a map of Arizona (\cref{FIG:Arizona}): at the highest zoom level when one can see the entire state, one will mainly see Interstates and other major highways; zooming in on Tucson, major city streets begin to appear, and at a fine zoom level of University of Arizona's campus, one can see small roads and pedestrian walkways.  One could do a similar representation on a map of collaboration networks as in \cite{fried2014maps}.

\begin{figure}[h!]
    \centering
    \includegraphics[width=0.3\textwidth]{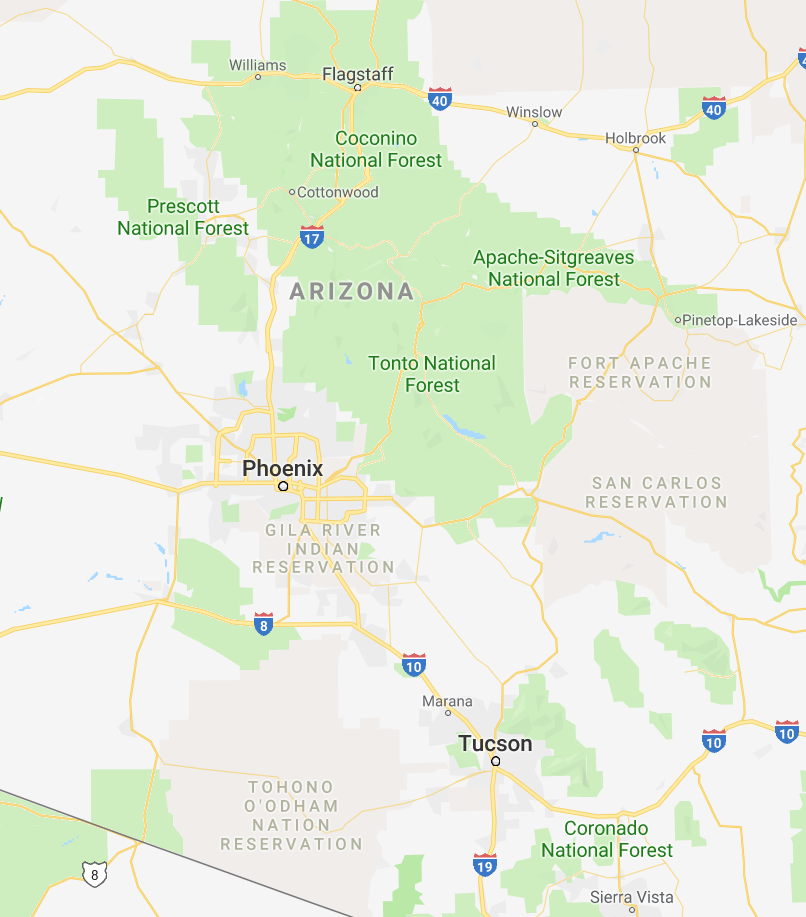}\;\;
    \includegraphics[width=0.28\textwidth]{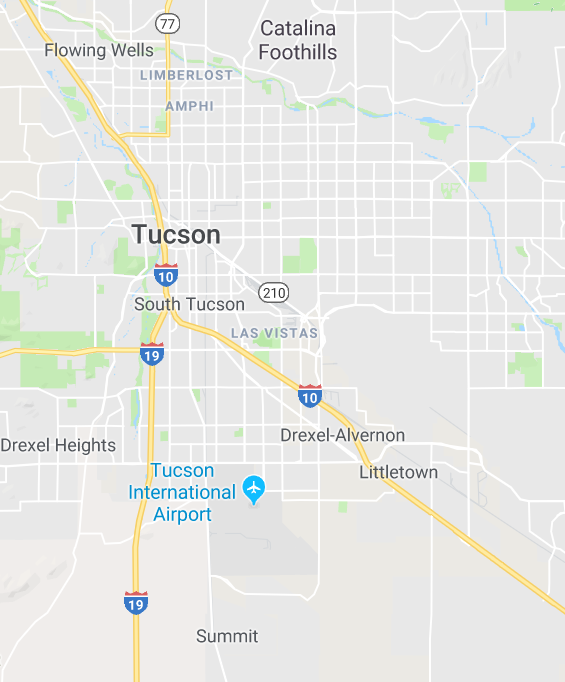}
    \;\;\includegraphics[width=0.3\textwidth]{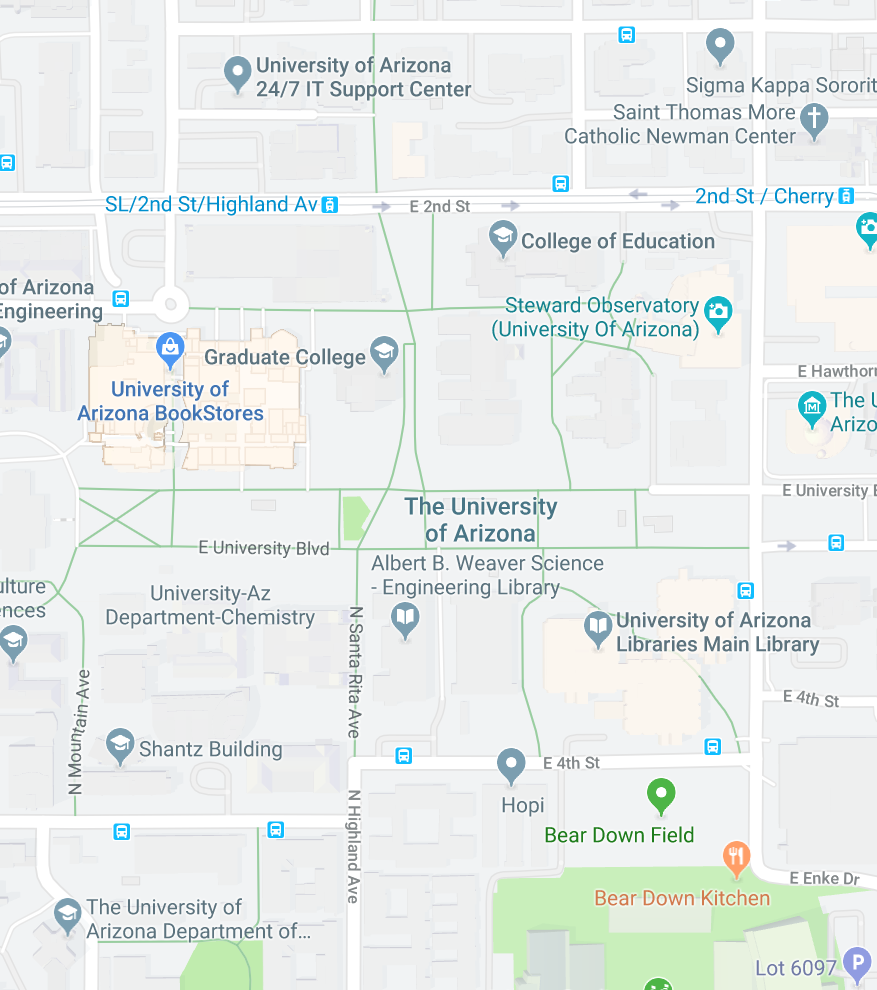}
    \caption{Three different zoom levels of a map of Arizona (\emph{Map data: Google}).}
    \label{FIG:Arizona}
\end{figure}

Here, we will consider multi-level constrained terminal sketches defined as follows. We denote $[\ell]=\{1,\dots,\ell\}$, and use $G_i\subset G_{i-1}$ to denote that $V_i\subset V_{i-1}$ and $E_i\subset E_{i-1}$.

\begin{definition}[Multi-level Constrained Terminal Sketch]\label{DEF:MultilevelConstrainedSketch}
Given an input graph $G$ which is to be sketched on $\ell$ levels, $\ell$ nested terminal sets $T_\ell\subset\dots\subset T_1\subset V$, and a set of constraints $\C$ (which will be the same for all levels), a sequence of graphs $G_\ell\subset\dots\subset G_1\subset G$ is a \emph{multi-level $\C$--terminal sketch of $G$ over $\{T_i\}_{i=1}^\ell$}  if $G_i\in \Schwartz(G,T_i,\C)$ for all $i\in[\ell]$.
\end{definition}

Thus for a multi-level sketch, we require that the subgraphs $G_i$ on each level are a $\C$--terminal sketch of the full graph $G$ over $T_i$.  Note that this does not require that $G_i$ is a terminal sketch of $G_{i-1}$ over $T_{i}$. 

It is natural that if an edge $e$ appears on multiple levels, then it should cost more than if it had appeared on a single level.  Thus, one possibility is that the \textit{cost} of an edge $e$ could be its base cost $c(1,e)$ times the number of levels it appears on.  We formulate the multi-level version of \cref{PROB:ConstrainedSketch} to include the most general definition of cost, but we will also consider the case with uniform costs. 

\begin{problem}[Minimum Cost Multi-level Constrained Terminal Sketch Problem (MLCTS)]\label{PROB:Multilevel}
Given a graph $G(V,E)$ which is to be sketched on $\ell$ levels, nested terminal sets $T_\ell\subset\dots\subset T_1\subset V$, a distance function $d:E \to \R_+$, set of constraints $\C$, and a level/edge cost function $c: [\ell]\times E\to \R_+ $, find a multi-level $\C$--terminal sketch of $G$ over $\{T_i\}_{i=1}^\ell$ whose total cost is minimized.  That is, if $y(e)$ is the largest level on which an edge $e$ appears in $\{G_i\}$, solve
\[
    \underset{\substack{G_\ell\subset\dots\subset G_1\\ G_i\in\Schwartz(G,T_i,\C)}}{\argmin}\; \textnormal{COST}(G_\ell,\dots,G_1) :=\underset{\substack{G_\ell\subset\dots\subset G_1\\ G_i\in\Schwartz(G,T_i,\C)}}{\argmin}\;\sum_{e\in E} c(y(e),e).
\] 
\end{problem}

The cost function $c(i,e)$ in \cref{PROB:Multilevel} represents the cost of including edge $e$ on levels $1, \ldots, i$ with the convention $c(0,e) = 0$. For example, if edge $e$ has $(c(1,e), c(2,e), c(3,e)) = (1,3,5)$, $e \in G_1, G_2$, and $e \not\in G_3$, then we pay a cost of 3 for including edge $e$. For a fixed edge, $c(i+1,e)-c(i,e)$ can be thought of as the cost of ``upgrading'' edge $e$ from level $i$ to level $i+1$.

An equivalent formulation of the MLCTS problem is in terms of \emph{grades of service}.  Consider a function $R: V \to \{0,1,\ldots,\ell\}$ such that $R(v)$ indicates the required grade of service of $v$ (equivalently, the maximum $i$ for which $v \in T_i$). The MLCTS problem can be equivalently stated as follows: compute a single subgraph $H \subset G$ by including edges $e$ of varying grade of service $y(e)$ such that, for all $i \in [\ell]$, the subgraph induced by the edge set $\{e: y(e) \ge i\}$ is a $\C$-terminal sketch of $G$ over $T_i = \{v: R(v) \ge i\}$.

In the sequel, we will consider the important special case of \emph{uniform} edge costs, where the cost of all edges uniformly scale across levels:

\begin{definition}\label{def:uniform}
A level/edge cost function $c:[\ell] \times E \to \R_+$ is \emph{uniform} if, for all $i \in [\ell]$ and $e \in E$, we have $c(i,e) = g(i) \cdot c(1,e)$, where $g(i)$ is a non-decreasing function with $g(1)=1$.
\end{definition}
We will use the phrase ``uniform edge costs'' when specifying an instance whose cost function is uniform.

\subsection{Merging of Single-Level Sketches}

One of the main focuses of this work is to determine how effective it is to solve single-level constrained terminal sketch problems to approximate the solution to the MLCTS problem; i.e., how well does solving \cref{PROB:ConstrainedSketch} for a subset of levels and combining these solutions approximate the solution to \cref{PROB:Multilevel}?  To achieve this, we need a notion of how to combine constrained terminal sketches.  We thus require that the set of constraints $\C$ admit a \textit{merging operator} such that
\[ \bigoplus:\Schwartz(G,T_1,\C)\times \Schwartz(G,T_2,\C)\to\Schwartz(G,T_1,\C) \]
for every $T_2\subset T_1\subset V$, and moreover, that $G_2\subset G_1\oplus G_2\subset G_1\cup G_2$.

\subsection{Examples of Constraints and Merging Operators}

If $\C$ is the single constraint ``is a tree", then \cref{PROB:ConstrainedSketch} is the \textit{Steiner tree} problem (ST). The multi-level version (\cref{PROB:Multilevel}) has been studied previously under various names including multi-level Steiner tree, priority Steiner tree, quality-of-service multicast tree; refer to \Cref{subsection:related} for related results. Note that if $\C=\emptyset$, then one still obtains the Steiner tree as the solution to \cref{PROB:ConstrainedSketch} due to the fact that $G'$ must be connected, and the Steiner tree is the minimum-cost connected subgraph which spans the prescribed terminals; however, this case does not constrain the search space to satisfy \cref{DEF:CTerminalSketch} at all. If $T=V$, the tree constraint is equivalent to finding a minimum spanning tree (MST) of $G$.

Note that the union of two Steiner trees need not be a tree. Hence the merging operator in this case requires some thought.  If $G_2$ and $G_1$ are Steiner trees over nested terminals $T_2\subset T_1$, a  merging operator can take the union $G_1\cup G_2$, then prune edges from $E_1 \setminus E_2$ to eliminate cycles. Note that $G_1\oplus G_2$ is not necessarily the minimum Steiner tree of $G$ over $T_1$; however, it is a candidate solution to \cref{PROB:Multilevel}.

If $\C$ requires that shortest path distances in $G'$ are preserved between terminals, i.e., $d_{G'}(u,v) = d_G(u,v)$ for $u,v\in T$. where $d_G(u,v)$ is the length of the shortest $u$--$v$ path in $G$ using the edge distances $d(\cdot)$, then \cref{PROB:ConstrainedSketch} is equivalent to finding a minimum cost \textit{terminal distance preserver}. 

More generally, $\C$ could require that $G'$ is a \textit{terminal (or subsetwise) spanner}. A spanner has as a parameter a distortion function $f:\R_+\to\R_+$ which satisfies $f(x)\geq x$, and the constraint is that $d_{G'}(u,v) \leq f(d_G(u,v))$ for $u,v\in T$.  Typical choices for distortion functions are $f(x) = tx$ (multiplicative $t$--spanner), $f(x)=x+\beta$ (additive $\beta$--spanner), or $f(x) = \alpha x+\beta$ (linear, or ($\alpha,\beta$)--spanner); see \cite{spannersurvey} for a survey of spanners.  For this choice of constraint $\C$, we call the problem the Multi-Level Graph Spanner problem (MLGS).  For clarity, the precise problem statement is as follows.

\begin{definition}\label{DEF:MLGS}[Multi-level Graph Spanner (MLGS) problem]
Given a graph $G(V,E)$ which is to be sketched on $\ell$ levels, nested terminal sets $T_\ell\subset\dots\subset T_1\subset V$, a distance function $d:E \to \R_+$, a spanner distortion function $f:\R_+\to\R_+$ satisfying $f(x)\geq x$, and a level/edge cost function $c: [\ell]\times E\to \R_+$, compute a minimum-cost sequence of spanners $G_{\ell} \subset \dots \subset G_1$, where $G_i$ is a terminal (subsetwise) spanner of $G$ over $T_i$ with distortion $f$ for $i\in[\ell]$.  That is, $d_{G_i}(u,v)\leq f(d_G(u,v))$ for all $u,v\in T_i$. The cost of a solution is defined as before: $\sum_{e\in E} c(y(e),e)$.
\end{definition}

In the case of terminal spanners with distortion $f$, a natural merging operation is $\bigoplus = \cup$, i.e., $G_1\bigoplus G_2 = (V_1\cup V_2,E_1\cup E_2)$.  It is straightforward to verify that $G_1 \bigoplus G_2$ is a terminal spanner over $T_1$ with distortion $f$.

Alternatively one could require that $G'$ be a spanner as well as constraining that $G'$ has diameter at most $D$ for some prescribed $D>0$ (termed \emph{low-diameter spanners}).  Here, the merging operation is the same as for spanners, i.e., to take union.

Another possible constraint is that $G'$ be a $k$--connected subgraph, which means that it must have the property that every pair of vertices in $T$ is connected in $G'$ by at least $k$ vertex-disjoint paths.   In this case, we have $\oplus=\cup$ as in the spanner case.

\subsection{Our Contributions}

For more precise statements of the main results, see \Cref{SEC:MainResults}; here we briefly explain the main contributions of this article.

We provide an overarching framework for computing multi-level constrained terminal sketches given access to a single-level solver (either an approximation algorithm or oracle).  Through this, we demonstrate upper bounds for the best approximations to the multi-level sketching problem (\cref{PROB:Multilevel}) via solutions to the single-level \cref{PROB:ConstrainedSketch}.  In particular, we show that with uniform edge costs, the MLCTS admits a 4-approximation if one is given access to an oracle which computes the corresponding minimum-cost constrained (single-level) terminal sketch.

Additionally, we provide a simple single-level approximation algorithm for the terminal (subsetwise) spanner problem and give upper bounds for the weight of the derived spanner which are state-of-the-art in terms of the tradeoff between lightness and time.  Utilizing this construction, we also give the first weight bounds for multi-level graph spanners.  As an additional conceptual contribution, our weight bounds for subsetwise spanners suggest a new definition of lightness of them in terms of the weight of the minimum-weight Steiner tree over the subset.

\subsection{Related Work}\label{subsection:related}

Multi-level or priority Steiner trees have been studied in \cite{MLST2018, 1288137, Chuzhoy2008}. Charikar et al.\cite{1288137} show that this problem with uniform edge costs admits an $O(1)$--approximation, and a $\min\{2 \ln|T_1|+2, \ell \alpha_{ST}\}$--approximation for nonuniform edge costs where $T_1\subset V$ is the set of terminals on level 1, and $\alpha_{ST} \approx 1.39$~\cite{Byrka2013} is an approximation ratio for the Steiner tree problem (by nonuniform costs here, we mean a cost function of the form $c(i,e)$ which cannot be factored as $c(i,e)=g(i)c(1,e)$). Chuzhoy et al.~\cite{Chuzhoy2008} show that this problem cannot be approximated with ratio $O(\log\log n)$ in polynomial time unless NP$\,\subseteq\,$DTIME$(n^{O(\log\log\log n)})$. The study of multi-level spanners was initiated in \cite{MLGS_proceeding}, but only achieved level-dependent bounds on the approximation algorithms for \cref{PROB:Multilevel}.  

The literature on spanners is vast, and we refer the reader to the survey \cite{spannersurvey} for an extensive reference list.  We highlight here that subsetwise distance preservers have been studied by Coppersmith and Elkin~\cite{coppersmith2006sparse} who showed that given an undirected weighted graph and a subset of size $O(n^\frac{1}{4})$, one can construct a linear size preserver in polynomial time. Cygan et al.~\cite{Cygan13} give polynomial time algorithms to compute subsetwise and pairwise additive spanners for unweighted graphs and show that there exists an additive pairwise spanner of size $O(n|T|^\frac{1}{2}{\log n}^{\frac{1}{4}})$ with $4 \log n$ additive stretch. They also show how to construct $O(n|T|^\frac{1}{2})$ size subsetwise additive 2--spanners. Abboud and Bodwin~\cite{Abboud16} improved that result by showing how to construct $O(n|T|^\frac{2}{3})$ size subsetwise additive 2--spanners. Kavitha~\cite{Kavitha2017}  shows that there is a polynomial time algorithm which constructs subsetwise spanners of size $\tilde{O}(n|T|^\frac{4}{7})$ and $O(n|T|^\frac{1}{2})$ for additive stretch 4 and 6, respectively. Kavitha~\cite{Kavitha2017} also shows that there exists an algorithm which computes a subsetwise $(1+\eps,4)$--spanner of size $O\left(n\sqrt{\frac{|T|\log n}{\eps}}\right)$. Bodwin and Williams~\cite{Bodwin:2016:BDP:2884435.2884496} give an upper bound on the size of subsetwise spanners with polynomial additive stretch factor. 

To the authors' knowledge, there are no existing guarantees in the literature for multiplicative subsetwise spanners except those of Klein~\cite{Klein06} who gives a polynomial time algorithm that computes a subsetwise multiplicative spanner of an edge weighted planar graph for a constant stretch factor with constant approximation ratio. On the other hand, for general graphs, there is a folklore upper bound of $O(|T|^2)$ on the size of multiplicative subsetwise spanners in the case that $|T|$ is polynomially larger than $\sqrt{n}$; the construction uses the algorithm of Elkin and Peleg~\cite{Elkin:2004:SCG:976327.984900}.  That this upper bound is relatively small indicates one reason for the lack of extensive study of subsetwise spanners as opposed to more general pairwise spanners; however, our simple \cref{ALG:subsetwise} yields new bounds on the size of subsetwise spanners in terms of the weight of the Steiner tree over the subset.

The hardness of approximation of multi-level spanners follows from the single level case. Peleg and Sch\"{a}ffer~\cite{doi:10.1002/jgt.3190130114} show that determining if there exists a $t$--spanner of $G$ with $m$ or fewer edges is NP--complete. Further, it is NP--hard to approximate the (unweighted) $t$--spanner problem for $t>2$ to within a factor of $O(\log n)$ even when restricted to bipartite graphs~\cite{Kortsarz2001}. Recently, Dinitz et al.~\cite{dinitz2016label} proved the super-logarithmic hardness of the basic $t$--spanner problem by showing that for every integer $k \geq 3$ and every constant $\epsilon >0$ it is hard to approximate an optimal $t$--spanner within a factor better than $2^{\log^{1-\epsilon}n/k}$.

\subsection{Layout}

\Cref{SEC:Notation} contains the notations and assumptions used throughout the sequel. The main results are stated in \Cref{SEC:MainResults}, and  \Cref{Sec:SingleLevel} discusses how single-level solvers can be used to give approximation algorithms for the MLCTS problem.  \Cref{SEC:Corollaries} contains the results for specific cases of constraints, with a focus on the terminal spanner problem.  \Cref{SEC:MultiLevelCorollaries} illustrates several multi-level results obtained using various single-level solvers from the literature.  Some numerical experiments for multi-level spanners are given in \Cref{SEC:Experiments}, and \Cref{SEC:Conclusion} is a brief conclusion to the paper.

\section{Notation and Assumptions}\label{SEC:Notation}

All graphs $G(V,E)$ are connected, undirected, and have a given distance function $d:E\to\R_+$ associated with them. The distance function here can be thought of as analogous to weights on the edges, but we use ``distance'' to better distinguish from the cost function on the edges. For $u,v\in V$, $d_G(u,v)$ denotes the length of the shortest $u$--$v$ path using distance function $d$.  Terminal sets will be denoted by $T\subset V$, and in the multi-level graph problems, we consider a generic cost function $c:[\ell]\times E\to\R_+$ to describe the cost of including edge $e$ on a given level and all levels below it.  

We denote by $\ST(G,T)$ a minimum Steiner tree of $G$ over terminals $T$.  $\C$ denotes a set of constraints in \cref{PROB:ConstrainedSketch} and subsequent problems.  For multiplicative and additive terminal spanners, we typically write terminal multiplicative $t$--spanner or terminal additive $\beta$--spanner (substituting other expressions for $t$ and $\beta$ frequently).

Lastly, we assume that given $G$ and $\C$, there exists a $\C$-terminal sketch of $G$ over any terminal set $T$, i.e., $\Schwartz(G, T, \C) \neq \emptyset$ for all $T \subset V$.

\section{Main Results}\label{SEC:MainResults}

Here we outline the main results presented in the paper.

\subsection{Multi-Level Sketches}

Our first result pertains to what happens when one has access to an oracle to solve the single-level constrained terminal sketch problem.

\begin{theorem}[cf. \cref{THM:RoundingUp,COR:OracleRounding}]\label{THM:Oracle}
Let $\C$ be a fixed set of constraints, $\ell\in\N$, and let $G(V,E)$ and $T_\ell\subset\dots\subset T_1\subset V$ be given, and assume the cost function $c:[\ell]\times E\to\R_+$ is uniform.  If an oracle is used to solve the minimum cost single-level constrained terminal sketch problem, then with at most $\log\ell$ queries to the oracle, an admissible solution to the MLCTS problem can be found whose cost is a constant (independent of $\ell$) away from that of the optimal solution.
\end{theorem}

Of course one drawback of \cref{THM:Oracle} is that using an ILP as the oracle could well require exponential time; however, this indicates that in general, the multi-level sketching problem is not significantly harder than the single-level one.

\subsection{Single-Level Spanners}

In speaking of using single-level solvers, when the constraint is that the subgraph be a spanner, our framework necessitates a solution to the \textit{terminal (subsetwise) spanner problem}.  Precious few algorithms are available to solve the subsetwise spanner problem, and so we propose new results for finding them in the single-level case.  The following is a sample of the results we obtain in \Cref{SEC:Corollaries} (to match the spanner literature, here we assume the cost of an edge is its weight since there is only a single level).

\begin{theorem}[cf. \cref{THM:subsetwise1}]\label{THM:Introsubsetwise1}
Let $G(V,E)$ be weighted, $k\in\N$, $\eps>0$, and $T\subset V$ be given.  There exists a polynomial time algorithm which runs in $O(|T|^{2+\frac1k})$ time and yields a terminal (subsetwise) multiplicative $(2k-1)(1+\eps)$--spanner $G'(V',E')$ of $G$ over $T$ satisfying
\[W(G')= O(|T|^\frac1k)  W(\ST(G,T)).\]
Moreover, the output of the algorithm satisfies
\[W(G') = O(|T|^\frac1k)\OPT \]
where $\OPT$ is the minimum weight of any $(2k-1)(1+\eps)$--spanner of $G$.
\end{theorem}

Even in the case $T=V$, this result is new and state-of-the-art in terms of both running time and the tradeoff between the stretch factor and the weight of the spanner.  The above theorem gives a new, more realistic notion of \textit{lightness} of terminal spanners as it is expressed in terms of the weight of the Steiner tree over the subset $T$ (bounds in the case $T=V$ are in terms of the minimum-weight spanning tree, which is the Steiner tree in that case).  
Moreover, this weight bound is for arbitrary graphs, whereas the only previous bound of this type was for planar graphs~\cite{Klein06}.

\subsection{Multi-Level Spanners}

Combining the results of \cref{THM:Oracle,THM:Introsubsetwise1}, we have the following theorem for multi-level graph spanners.

\begin{theorem}[cf. \cref{COR:MLMultSpanner}]
Let $G(V,E)$, $k\in\N$, $\eps>0$, and $T_\ell\subset\dots\subset T_1\subset V$ be given, and assume the cost function $c:[\ell]\times E\to\R_+$ is uniform. There exists a polynomial time algorithm which runs in $O(|T_1|^{2+\frac1k+\eps})$ time and yields an $O(|T_1|^\frac1k)$--approximation to the MLCTS problem when the constraint requires computing a multiplicative $(2k-1)(1+\eps)$--spanner.
\end{theorem}

\section{Solving on a Subset of Levels}\label{Sec:SingleLevel}

Here, we analyze how to find an approximation to the optimal solution to the MLCTS problem when using single level solvers; for simplicity we will call these multi-level sketches.  We first suppose that we may solve (exactly or approximately) the problem on any subset of levels, and we discuss how one can determine \textit{which} levels one should use.

\subsection{Rounding Up Single-Level Solutions}\label{SEC:Quantization}

We begin by discussing solutions obtained from \textit{promoting}, or \textit{rounding} solutions from a given level up to a higher level.  This is similar to the approach of \cite{MLST2018} for Steiner trees and is based on a rounding algorithm of Charikar et al.~\cite{1288137}.  The analysis here shows that a similar scheme works for arbitrary constrained sketches we consider here as well as for more general cost functions.

A simple case of the algorithm may be summarized as follows: compute single-level sketches on dyadic levels ($1, 2, 4, \dots, 2^{\ceiling{\log_2\ell}}$), and carry the solutions up on intermediate levels; for instance, the sketch from level 4 is used as the sketch on levels 5, 6, and 7. These subgraphs are then merged in descending order from the top down to provide a feasible multi-level solution to the MLCTS problem.  This is the algorithm used for multi-level Steiner trees in \cite{MLST2018}, but the analysis there does not depend on the computed subgraphs being trees, so we can immediately state that the MLCTS problem over constraints $\C$ and uniform edge costs admits a 4-approximation if one is given access to an oracle to solve the corresponding \cref{PROB:ConstrainedSketch}.




We assume the cost of an edge $e$ on a given level $i$ can be decoupled as $c(i,e) = g(i) c(1,e)$ for all $e$, where $g:[\ell]\to\R_+$ is a cost scaling function.  For example, if $g(i) = i$, then the cost function $c(i,e)$ is linear.  As we consider here the case that $g$ is an arbitrary nondecreasing function, the cost model is more general than the special case $g(i)=i$ considered in~\cite{MLGS_proceeding}.

Here, we consider a general rounding set $\Q=\{1=i_1,\dots,i_m\}\subset[\ell]$. For each $i_j\in\mathcal{Q}$, compute an admissible sketch over terminals $T_{i_j}$, and let $H_{i_1}, \dots, H_{i_m}$ be the subgraphs returned. Then for $i=1,\dots,\ell$ we will set $G_i$ as follows:

\begin{equation}\label{eq:g_i}
G_i = \bigoplus_{j=k}^m H_{i_j},\quad i_k\leq i<i_{k+1}. 
\end{equation}

\noindent where $i_{m+1}:=\ell+1$. Furthermore, we use the convention that $\oplus_{j=1}^n H_j$ merges in descending order, e.g., $\oplus_{j=1}^3 H_j = H_1\oplus (H_2\oplus H_3)$. In other words, the graph on level $i$ where $i_k\leq i<i_{k+1}$, is the {\em merging} (in the sense of the operator $\oplus$) of all computed sketches $H_{i_j}$ on higher levels, as well as the computed sketch $H_{i_k}$. For example, if $\ell=6$ and $\mathcal{Q} = \{1,4,6\}$, then $G_6 = H_6$, $G_4 = G_5 = H_4 \oplus H_6$, and $G_1 = G_2 = G_3 = H_1 \oplus (H_4 \oplus H_6)$. 

\begin{algorithm}[h!]
\caption{Rounding Approximation Algorithm for MLCTS($G,T_1,\dots,T_\ell,\C,\mathcal{Q}$)} \label{ALG:composite}
\begin{algorithmic}
\For{$j=1,\dots,m$}
   \State $H_{i_j} \gets$ an $s(j)$--approximation to the minimum-cost $\C$--terminal sketch of $G$ over $T_{i_j}$ 
\EndFor
\For{$i=1,\dots,\ell$}
    
    \State $\displaystyle G_i = \bigoplus_{j=k}^m H_{i_j},\quad i_k\leq i<i_{k+1}$ (as in \cref{eq:g_i})
    
\EndFor
\State\Return $G_1,\dots,G_\ell$
\end{algorithmic}
\end{algorithm}

To analyze the approximation guarantees of \cref{ALG:composite}, suppose that $A=A(g,\Q)\geq1$ is the minimal possible constant such that
\begin{equation}\label{EQN:Gamma}
g(i_{k+1}-1)\leq Ag(i_k),\quad k\in[m],
\end{equation} 
and let $B=B(g,\Q)\geq1$ be the minimal constant such that 
\begin{equation}\label{EQN:q}
\sum_{j=1}^k g(i_j)\leq B g(i_k),\quad k\in[m].
\end{equation}
Note that the inequality \cref{EQN:q} is equivalent to $\sum_{j=1}^kc(i_j,e)\leq Bc(i_k,e)$ via the assumption on the form of $c$.  Note this is no restriction as such constants always exist -- indeed one may take $A:=\max_{k\in[m]}g(i_{k+1}-1)/g(i_k)$ and $B:=\max_{k\in[m]}\sum_{j=1}^k g(i_j)/g(i_k)$.  Do note however that if the level cost function $g$ increases rapidly, then $A$ can be quite large.

\begin{theorem}\label{THM:RoundingUp}
Given a graph $G(V,E)$, terminal sets $T_\ell\subset\dots\subset T_1\subset V$, a cost function $c:[\ell]\times E\to\R_+$ of the form $c(i,e)=g(i)w(e)$ for some level cost function $g:[\ell]\to\R_+$, and a rounding set $\mathcal{Q}=\{1=i_1,\dots,i_m\}$ which satisfies conditions \cref{EQN:Gamma} and \cref{EQN:q}, \cref{ALG:composite} yields an $ABs$--approximation to the MLCTS problem, where $s=\max_k\{s(k)\}$.
\end{theorem}
\begin{proof}
Consider a candidate solution to the MLCTS problem using $\Q=\{1=i_1,\ldots,i_m\}$  where edges of grade of service $i$ with $i_k\leq i<i_{k+1}$ are assigned levels $q(i)=i_{k+1}-1\in\{i_2-1,i_3-1,\ldots,\ell\}$.  For any grade or service assignment function $y:E\to[\ell]$, \cref{EQN:Gamma} implies that the cost of the solution restricted to $\Q$ satisfies
\begin{equation}\label{EQN:CostQ}
    \sum_{e\in E}c(q(y(e)),e)  = \sum_{e\in E} g(q(y(e)))c(1,e) \leq A\sum_{e\in E}g(y(e))c(1,e) = A\sum_{e\in E} c(y(e),e).
\end{equation}
Let $\OPT_\Q$ denote the cost of the optimal solution to the MLCTS problem using $\Q$.  Minimizing the left-hand side of \cref{EQN:CostQ} over all choices of $y$ yields $\OPT_\Q\leq A\sum_{e\in E} c(y(e),e)$ for any $y$.  Thus minimizing the right-hand side of \cref{EQN:CostQ} over $y$ as well yields $\OPT_\Q\leq A\OPT$, where $\OPT$ is the cost of the optimal solution to the original MLCTS problem (i.e., not depending on the set $\Q$).  

Next we show that the output of \cref{ALG:composite} has cost bounded above by $Bs\OPT_\Q$, which combined with the previous analysis yields the desired approximation bound.  To see this, note that the output of algorithm gives an $s$--approximation to the minimum-cost solution to the single-level sketching problem on levels $i_1,\dots,i_m$.  If these costs are $C_{i_j}$, and $\MIN_{i_j}$ is the minimum cost of a sketch on level $i_j$ (according to cost $c(1,e)$), then $C_{i_j}\leq s\MIN_{i_j}$.  

Consider a level assignment $y:E\to[\ell]$ corresponding to a solution to the rounded problem whose cost is $\OPT_\Q$.  Then we have that
\[\OPT_\Q = \sum_{e\in E}g(y(e))c(1,e) = \sum_{k=1}^m g(i_k)\sum_{e\in E}\mathbbm{1}_{\{y(e) = i_k\}}c(1,e), \]
where $\mathbbm{1}_S$ denotes the indicator function of the set $S$. The cost of a solution given by \cref{ALG:composite} is at most $s\sum_{k=1}^m g(i_k)\MIN_{i_k}$.

Now, suppose we take the solution with cost $\OPT_\Q$ and duplicate an edge on level $i_k$ on all levels below it in $\Q$: $i_{k-1},\dots,i_1$.  The cost of this is then \[ \sum_{k=1}^m g(i_k)\sum_{e\in E}\mathbbm{1}_{\{y(e)\leq i_k\}}c(1,e),\]
which is at least \[ \sum_{k=1}^m g(i_k)\MIN_{i_k}\] since the former contains feasible solutions to the single-level sketch at each level $i_k$.  Putting it all together, using \cref{EQN:q} we have
\begin{align*}
 s\sum_{k=1}^mg(i_k)\MIN_{i_k} &\leq s\sum_{k=1}^m g(i_k)\sum_{e\in E}\mathbbm{1}_{\{y(e)\leq i_k\}}c(1,e)\\
 & = s\sum_{k=1}^m g(i_k) \sum_{j=1}^k \sum_{e\in E}\mathbbm{1}_{\{y(e)= i_j\}}c(1,e)\\
 & = s\sum_{j=1}^m \left[\sum_{k=1}^j g(i_k)\right]\sum_{e\in E}\mathbbm{1}_{\{y(e)=i_j\}}c(1,e)\\
  & \leq Bs\sum_{j=1}^m g(i_j)\sum_{e\in E}\mathbbm{1}_{\{y(e)=i_j\}}c(1,e)\\
  & = Bs\OPT_\Q.
\end{align*}
Therefore, the solution output by \cref{ALG:composite} has cost at most $Bs\OPT_\Q$ which is at most $ABs\OPT$.
\end{proof}

Note that the proof of \cref{THM:RoundingUp} is independent on the type of sketch desired, and hence the algorithm is quite flexible.  A similar approach was used to approximate minimum-cost multi-level Steiner trees in~\cite{MLST2018}, but the above analysis shows this approach also works for spanners and $k$--connected subgraphs, for example. 

\begin{remark}
Note that we have presented here the idea of upgrading, or rounding up \textit{solutions} of single-level sketches, but we could equally well cast the algorithm in terms of upgrading or rounding up \textit{terminals}, e.g., assigning a terminal rate $16$ rather than its original rate of $10$.  This is the notion used by Charikar et al.~\cite{KORTSARZ1994222}, and we simply remark here that if \cref{ALG:composite} is modified in this language, then the approximation guarantee of \cref{THM:RoundingUp} is the same as long as the definition of constants $A$ and $B$ are suitably modified.  In fact, the solutions produced will also be the same.
\end{remark}

\subsection{Examples}\label{SEC:RoundingExamples}

The level cost function $g$ could take many forms.  For simplicity of examples here, we take $g(i)=i$, i.e., a linear cost growth along the levels.

Several natural options suggest themselves for quantization functions. Following Charikar et al.~\cite{1288137}, we may take $\Q = \{1,2,\dots,2^{\ceiling{\log_2\ell}}\}$ and round each level up to the nearest power of 2.  In this case, using an oracle to compute the terminal spanner at each level yields a $4$--approximation to the MLCTS problem for a multiplicative $t$--spanner.  The same approximation holds in this case for multi-level Steiner trees~\cite{MLST2018}.  Indeed,  $2^{i+1}-1 \leq 2\cdot2^i$, whence we may choose $A=2$, and if edge $e$ gets its rate rounded to $2^j$, then
\[\sum_{j=1}^k i_j = \sum_{j=1}^k 2^j \leq 2\cdot2^k, \] whence $B=2$.  This yields the following corollary which improves on Theorems 1 and 3 in~\cite{MLGS_proceeding}, as the approximation ratio here is independent of the number of levels.

\begin{corollary}\label{COR:OracleRounding}
Let $\C$ be a fixed set of constraints, $\ell\in\N$, let $G(V,E)$ and $T_\ell\subset\dots\subset T_1\subset V$ be given, and let the rounding set be $\Q=\{1,2,\dots,2^{\ceiling{\log_2\ell}}\}$.  Assuming uniform edge costs with $g(i)=i$, \cref{ALG:composite} produces a $4s$--approximation to the MLCTS problem.  In particular, using an oracle as a subroutine produces a $4$--approximation.

Moreover, a $4$--approximation to the minimum-cost multi-level multiplicative graph spanner problem may be obtained by using an oracle to solve the single-level spanner problem.
\end{corollary}
Note that the use of an oracle such as the ILP given in \cite{MLGS_proceeding} is time-intensive; this issue will be addressed in the sequel.

It is of interest to note that choosing the rounding set to be $\Q=\{1,b,\dots,b^{\ceiling{\log_bi}}\}$ for some other base $b>0$ does not improve the approximation ratio.  In this case, one can show that $A=b$, and $B$ is approximately $\frac{b}{b-1}$, and thus the approximation ratio is $\frac{b^2}{b-1}s$ which is minimized when $b=2$.

Using a coarser quantizer instead yields a worse approximation.  Consider the coarsest quantizer which sets $\Q=\{1\}$, with $g(i)=i$ as before.  In this case, $A=\ell$, and $B=1$, which means that the best one can do is an $\ell s$--approximation.  This approach corresponds to the \textit{Bottom Up} approach described in~\cite{MLGS_proceeding}.

If no rounding is done (i.e., $\Q=[\ell]$), then one computes a sketch at each level and merges them together according to the operation $\oplus$ going down the levels.  This can be considered a \textit{Top Down} approach to the problem, and yields an upper bound on the approximation ratio of $\frac{\ell+1}{2}s$. 

Regard that the approximation ratios above for both the top down and bottom up approaches were shown to be tight in the case that an oracle is used on a subset of levels (the case $s=1$) \cite{MLGS_proceeding}.

\subsection{Approximation Bounds Given $\Q$}\label{SEC:MINBounds}

The previous subsections described a general framework for level quantization which can yield a good approximation to the multi-level sketching problem.  In particular, for a dyadic rounding scheme, we showed that \cref{ALG:composite} yields a $4$--approximation to the MLCTS problem with uniform edge costs regardless of the constraints $\C$.  In this subsection, we illustrate some other approximation bounds that can be obtained for a prescribed rounding set $\Q$. For all of the results stated in the remainder of this section, we assume that we are given a set of constraints $\C$, a graph $G(V,E)$, a set of nested terminals $T_\ell\subset\dots\subset T_1\subset V$, and that an oracle is used as a subroutine in \cref{ALG:composite}.  If an approximation algorithm is used for the single-level solver, a factor of $s$ as in \cref{THM:RoundingUp} should be added to any approximation bounds.

First, let us consider an upper bound on the cost of multi-level sketches based upon the minimum-cost solution to the single level sketching problems. To wit, let $\MIN_i$ denote the minimum cost of a sketch of $G$ over $T_i$ using costs $c(1,e)$ for each edge $e$; i.e., \[\MIN_i := \min_{G_i\in\Schwartz(G,T_i,\C)}\;\sum_{e\in E_i}c(1,e).\]  Note that the cost above does not depend on the number of levels, as we are only solving the single-level sketch problem over different terminal sets. As $T_{\ell} \subset \dots \subset T_1$, we have $\MIN_{\ell} \le \dots \le \MIN_1$.  Given a choice of rounding set $\Q$, denote by $\COST_\Q$ the cost of the output of \cref{ALG:composite}. Then the following holds.

\begin{lemma}\label{lemma:cmp}
For any rounding set $\mathcal{Q}=\{1=i_1,\dots,i_m\}\subset[\ell]$, using an oracle as a subroutine, \cref{ALG:composite} produces an output with cost \[\COST_\Q \le f(Q,\{\MIN_i\}):= \sum_{k=1}^m g(i_{k+1}-1)\MIN_{i_k},\] where $i_{m+1} :=\ell+1$.
\end{lemma}
\begin{proof}
The single level sketch on terminals of level $i_k$ is promoted to level $i_{k+1}-1$, i.e., edges are assigned grade of service $i_{k+1}-1$. Each subgraph $H_{i_k}$ then has cost $g(i_{k+1}-1)\MIN_{i_k}$ and is computed by a single-level oracle. The rest of the proof follows from the merging bound and the fact that $G_1\oplus G_2\subset G_1\cup G_2$. 

\end{proof}

The utility of the above approximation may be limited as it requires knowledge of the optimal cost of each single-level sketch.  However, we may use it to state a general approximation bound for any rounding set, independent of the cost of a single-level solution to \cref{PROB:ConstrainedSketch}.

\begin{lemma}\label{lemma:QappRatio}
For any rounding set $\mathcal{Q}=\{1=i_1,\dots,i_m\}$, using an oracle as a subroutine, \cref{ALG:composite} yields the approximation guarantee \[\COST_\Q\leq \min_{1\leq h\leq m}\frac{\sum_{k=1}^h g(i_{k+1} - 1)}{i_h}\;\OPT.\]
\end{lemma}

\begin{proof}
Applying \cref{lemma:cmp} for a particular $\mathcal{Q}$, we have
\begin{equation}\label{Eq: R2OPT}
    \frac{\COST_\Q}{\OPT}
    \le \frac{\sum_{k=1}^m g(i_{k+1} - 1)\MIN_{i_k}}{\sum_{i=1}^\ell \MIN_i},
\end{equation}
where we have used the observation that $\OPT\ge\sum_{i=1}^\ell \MIN_i.$ Therefore, the right-hand side of the above inequality provides an approximation guarantee assuming knowledge of $\MIN_1,\dots,\MIN_\ell$. We can also find a generic bound by looking at the worst case scenario for $\MIN_1,\dots,\MIN_\ell$. Without loss of generality, we may assume that $\sum_{i=1}^{\ell} \MIN_i = 1$, so that $\OPT \ge 1$. Since $g(i)$ is an increasing function, the worst case of level costs will be of the form $\MIN_1=\MIN_2=\cdots=\MIN_h=\frac{1}{h}$ and $\MIN_{h+1}=\cdots=\MIN_\ell=0$ for some $h$ (see, e.g., \cite[Lemma 2.4]{MLST2018}). Therefore, the general approximation guarantee (regardless of costs $\MIN_i$ of the subset sketches of each level) is
\[\min_{1\leq h\leq m}\frac{\sum_{k=1}^h g(i_{k+1} - 1)}{i_h}.\]
\end{proof}

Now we note that \cref{lemma:QappRatio} actually provides a better approximation guarantee than \cref{THM:RoundingUp} in some cases.
\begin{proposition}\label{PROP:QappRatio}
For any rounding set $\Q=\{1=i_1,\dots,i_m\}$, using an oracle as a subroutine, \cref{ALG:composite} yields the approximation guarantee
\[\COST_\Q \leq AB\min_{1\leq h\leq m}\frac{g(i_{h})}{i_h}\;\OPT. \]
\end{proposition}
\begin{proof}
Combine \cref{EQN:Gamma}, \cref{EQN:q}, and \cref{lemma:QappRatio}. 
\end{proof}

\subsection{Minimum-Cost Rounding Solution}\label{SEC:Composite}


The previous subsection discussed approximation guarantees when a rounding set $\Q$ is prescribed beforehand. 
To find the best possible approximation to the MLCTS problem via single-level solutions, we can run \cref{ALG:composite} for all $2^{\ell-1}$ subsets $\mathcal{Q}\subset[\ell]$ containing 1, and then pick the one with the lowest cost. This strategy yields an approximation algorithm with constant ratio. The approximation algorithm requires $\ell$ sketch computations followed by finding the minimum cost over $2^{\ell-1}$ solutions.  Consequently, this method is more costly to implement than simply rounding up solutions as before, but it provides the lowest cost and the best guarantees possible using single-level solvers.  In the terminology of \cite{MLST2018}, we call this the Composite algorithm.

\begin{algorithm}[H]
\caption{Composite($G,T_1,\dots,T_\ell,\C$)}\label{ALG:compositeBestQ}
\begin{algorithmic}
\For{all $\Q \subset [\ell]$ with $\ell \in \Q$}
   \State $G_1,\dots,G_\ell \gets$ MLCTS($G,T_1,\dots,T_\ell,\C,\Q$) via \cref{ALG:composite}
   \State $\COST_\Q=\COST(G_1,\dots,G_\ell)$
\EndFor
\State\Return $\argmin\;\COST_\Q$
\end{algorithmic}
\end{algorithm}

\begin{proposition}\label{Prop:compositeBestQRatio}
The following linear program (LP) yields an approximation guarantee of the composite algorithm (\cref{ALG:compositeBestQ}):

\begin{align*}
    \max_{\{t,y_1,\dots,y_\ell\}} t\\
    \text{ subject to } \hspace{5ex}
    &t<\sum_{k=1}^m g(i_{k+1} - 1)y_{i_k},\text{ for all }\Q=\{1=i_1,\dots,i_m\}\subset[\ell]\\
    &\sum_{i=1}^{\ell} y_i = 1, y_\ell\leq\dots\leq y_1
\end{align*}
\end{proposition}

\begin{proof}
The composite \cref{ALG:compositeBestQ}\ picks the smallest cost returned from \cref{ALG:composite} over all possible rounding sets $\Q$. A general approximation for \cref{ALG:compositeBestQ} follows from:
$$\frac{\min\COST_\Q}{\OPT}\leq\frac{\min_Q f(\Q,\{\MIN_i\})}{\sum_i^\ell \MIN_i}\leq\max_{\{y_i\},\sum y_i=1,y_i\leq y_{i+1}}\min_\Q f(\Q,\{y_i\}).$$

The linear program in \cref{Prop:compositeBestQRatio} solves the optimization problem on the right-hand side of the above inequality. 
\end{proof}

This optimization problem suggested above is similar to that in~\cite{MLST2018}.  In particular, for linear costs ($g(i)=i$) and $\ell\leq 100$, the solution returned by the composite algorithm has cost no worse than $2.351\, \OPT$ assuming an oracle that computes sketches optimally. 

\subsection{Finding the Best Rounding Set}\label{SEC:BestComposite}

Suppose having computed all single-level solutions, we are interested to find what $\Q$ would provide the best multi-level solution without computing the cost for all $2^{\ell-1}$ admissible subsets of $[\ell]$ as in \cref{ALG:compositeBestQ}.
To do so, we can use the inequality of $\COST_Q\leq f(Q,\{\MIN_i\}):=\sum_{k=1}^m g(i_{k+1}-1)\MIN_{i_k}$, and instead find the subset $\Q^*$ that minimizes $f(Q,\{\MIN_i\})$ and then feed this newly found $\Q^*$ to \cref{ALG:composite}. Interestingly, even though most likely $\COST_{\Q^*}> \min_\Q\COST_\Q$ (i.e., the cost returned from \cref{ALG:composite} using $\Q^*$ is not necessarily the lowest possible cost), using $\Q^*$ in \cref{ALG:composite} yields the same approximation guarantee as \cref{Prop:compositeBestQRatio}.  This fact is recorded below in \cref{PROP:SameGuarantee}.

To state this algorithm formally, we first formulate a minimization problem as follows. Define binary variables $\theta_{ij}$ associated with a given $\Q$ such that $\theta_{i_k i_{k+1}}=1$ and $0$ otherwise. For example if $\ell=3$ and $\Q=\{1,3\}$, $\theta_{13}=\theta_{34}=1$ and $\theta_{12}=\theta_{14}=\theta_{23}=\theta_{24}=0$.

\begin{proposition}\label{Lemma: Qnew}
 Given a vector $\mathbf{y}=[y_1,\dots,y_{\ell}]$, the choice of $\Q^*=\{i:\theta_{ij}=1\}$ from the following ILP minimizes $f(\Q,\{y_i\})=\sum_{k=1}^{m} g(i_{k+1}-1)y_{i_k}$, where $i_k$ is the $k-$th smallest element of $\Q^*$ and $i_{m+1}=\ell+1$.

\begin{align*}
    \min_{\{\theta_{ij}\}} \sum_{i=1}^{\ell}\sum_{j=i+1}^{\ell+1} g(j-1)\theta_{ij}y_i\\
    \text{ subject to } \hspace{4ex}
     \sum_{j>i} \theta_{ij} &\le 1,  \quad \sum_{i<j} \theta_{ij} \le 1, & &  i \in [\ell],j \in \{2,\dots,\ell+1\}\\
     \sum_{i<k} \theta_{ik} &=  \sum_{j>k} \theta_{kj} & &  k\in\{2,\dots,\ell\}\\
     \sum_{1<j} \theta_{1j} &= 1, \quad
     \sum_{i<\ell+1}\theta_{i(\ell+1)}= 1  \\
\end{align*}

\end{proposition}

\begin{proof}
Using the indicator variables, the objective function can be expressed as\\ $\sum_{i=1}^{\ell}\sum_{j=i+1}^{\ell+1} g(j-1)\theta_{ij}y_i$ because $\theta_{i_k,i_{k+1}}=1$ and the other $\theta_{ij}$'s are zero. In the above formulation, the first constraint indicates that for every given $i$ or $j$, at most one $\theta_{ij}$ is equal to one. The second constraint indicates that for a given $k$, if $\theta_{ik}=1$ for some $i$, then there is also a $j$ such that $\theta_{kj}=1$. In other words, the result determines a proper choice of levels by ensuring continuity of $[i,j]$ intervals. The final constraint guarantees that $1\in\Q^*$.
\end{proof}

The above lemma suggests the following algorithm for computing an approximation to the minimum-cost multi-level sketch by first choosing the best rounding set from the previous proposition.

\begin{algorithm}[H]
\caption{Best Rounding Algorithm($G,T_1,\dots,T_\ell,\C$)}\label{ALG:compositeQstar}
\begin{algorithmic}
\For{$i=1,\dots,\ell$}
   \State $\MIN_i \gets$ cost of optimal $\C$--terminal sketch of $G$ over $T_i$ using an oracle
\EndFor
\State Setting $\mathbf{y}=[\MIN_1,\dots,\MIN_\ell]$ find $Q^*$ from \cref{Lemma: Qnew}.
\State\Return $G_1,\dots,G_\ell =$ MLCTS($G,T_1,\dots,T_\ell,\mathcal{C},Q^*)$ from \cref{ALG:composite}
\end{algorithmic}
\end{algorithm}

\begin{proposition}\label{PROP:SameGuarantee}
Let $\COST_{\Q^*}$ be the cost of the output of \cref{ALG:compositeQstar}, $\OPT$ be the cost of the optimal solution to the MLCTS problem, and $t_\ell$ be the approximation guarantee of the minimum-cost \cref{ALG:compositeBestQ} according to \cref{Prop:compositeBestQRatio}.  Then,
\[\COST_{\Q^*}\leq t_\ell\;\OPT.\]
\end{proposition}
\begin{proof}
Combining \cref{Prop:compositeBestQRatio} with \cref{lemma:cmp}, we have \[\frac{\COST_{\Q^*}}{\OPT}\leq\frac{\min_Q f(\Q,\{\MIN_i\})}{\sum_i^\ell \MIN_i}\leq t_\ell.\]
\end{proof}

Note that while only $\ell$ oracle queries are used in both \cref{ALG:compositeBestQ,ALG:compositeQstar}, the former requires computing the cost of a multi-level sketch for $2^{\ell-1}$ rounding sets, while the latter requires the solution to the ILP described in \cref{Lemma: Qnew}. 






\section{Single Level Approximation Algorithms}\label{SEC:Corollaries}

Solving a single level ILP takes much less time than solving a multi-level ILP, especially as the number of levels increases. This fact was one of the motivations behind the rounding algorithm. If a single-level ILP is used as an oracle subroutine, then a constant approximation ratio is obtained, but at the expense of the subroutine being exponential time.  

It is natural to consider what happens when an approximate algorithm for single-level sketch problem is used as the approximation algorithm in \cref{ALG:composite} since \cref{THM:RoundingUp} implies this still gives a good approximation to the multi-level sketching problem.

\subsection{Steiner Trees}

Results for Steiner trees are well-known. A very simple $2$--approximation is due to \cite{Gilbert1968} and the current best approximation algorithm guarantees an approximation ratio of $\rho\approx 1.39$ \cite{Byrka2013}. However, the Steiner tree problem is APX--hard \cite{Bern1989}, and it is NP--hard to approximate the problem within a factor of $96/95$~\cite{Chlebnik2008}.

\subsection{Spanners}\label{SEC:Spanners}

Unfortunately, there are few algorithms that provably approximate the subsetwise multiplicative spanner problem for generic graphs. Here, we give a new and simple algorithm for computing a subsetwise spanner with general distortion function $f$, which is among the first results of its kind to our knowledge. We will use an $f$--spanner subroutine in this algorithm which must work for weighted graphs. Note that most spanner algorithms for weighted graphs are usually of the multiplicative type. The idea is to apply an $f$--spanner subroutine on the metric closure of the graph over $T$ defined below.

\begin{definition}[Terminal Metric Closure]
Given a graph $G(V,E)$ and a set of terminals $T\subset V$, the \emph{terminal metric closure of $G$ over $T$} is defined as the complete weighted graph on $|T|$ vertices with the weight of edge $(u,v)$ given by $d_G(u,v)$.
\end{definition}

\Cref{ALG:subsetwise} describes a subsetwise spanner construction that provides a novel approximation guarantee.

\begin{algorithm}[H]\caption{Terminal (subsetwise) $f$--Spanner($G,T,f$)}\label{ALG:subsetwise}
\begin{algorithmic}
\State $ \tilde{G} \gets $ metric closure of $G$ over $T$
\State $\tilde{G}' \gets$ spanner of $\tilde{G}$ with distortion $f$
\State $E' \gets$ edges in $G$ corresponding to $\tilde{G}'$
\State \Return $E'$
\end{algorithmic}
\end{algorithm}

Note that the spanner obtained in the second step can come from any known algorithm which computes a spanner for a weighted graph, which adds an element of flexibility.  This also allows us to obtain some new results on the weight of multiplicative subsetwise spanners.  First let us note that \cref{ALG:subsetwise} yields a subsetwise spanner.

\begin{lemma}
Given $G(V,E)$, $T\subset V$, and a distortion function $f$, \cref{ALG:subsetwise} yields a subsetwise spanner $G'$ of $G$ over $T$ with distortion $f$.
\end{lemma}
\begin{proof}
Let $u,v\in T$.  Since $\tilde{G}'$ is a spanner of $\tilde{G}$ with distortion $f$, we necessarily have $d_{\tilde{G}'}(u,v) \le f(d_{\tilde{G}}(u,v)) = f(d_G(u,v))$. Since each edge in $\tilde{G}'$ induces a shortest path in $G'$, we have $d_{G'}(u,v) \le d_{\tilde{G}'}(u,v)\leq f(d_G(u,v))$ as desired.
\end{proof}

It is known (see~\cite{Alstrup2017}) that if $k$ is a positive integer and $\eps>0$, then one can construct a $(2k-1)(1+\eps)$--spanner for a given graph (over all vertices, not a subset) in $O(n^{2+\frac1k + \eps})$ time, which has weight $O(n^\frac1kW(\MST(G))$, where $\MST(G)$ is a minimum spanning tree of $G$.  Using this in line 2 of \cref{ALG:subsetwise}, we can conclude the following.

\begin{theorem}\label{THM:subsetwise1}
Let $G(V,E)$, $k\in\N$, $\eps>0$, and $T\subset V$ be given.  Using the spanner construction of~\cite{Alstrup2017} as a subroutine, \cref{ALG:subsetwise} yields a subsetwise multiplicative $(2k-1)(1+\eps)$--spanner $G'(V',E')$ of $G$ in $O(|T|^{2+\frac1k+\eps})$ time with total weight
 \[W(G')= O(|T|^\frac1k)  W(\ST(G,T)).\]
Moreover,\[W(G') = O(|T|^\frac1k) \OPT.\]
\end{theorem}
\begin{proof}
Let $\tilde{G}$ be the metric closure of $G$ over $T$, and let $G^*(V^*,E^*)$ be the minimum weight subsetwise $f$--spanner for $G$, where $f$ is any distortion function.
Note that $G^*$ must contain a tree, $G^*_0$, which spans $T$, whence
\begin{equation}\label{EQ:STBound} W(G^*) \geq W(G^*_0)\geq W(\ST(G,T))\end{equation}
by definition of Steiner trees.
From the approximation result for Steiner trees (see~\cite{bang2004}) we have 
\begin{equation}\label{EQ:MSTBound1} W(\MST(\tilde{G}))\leq 2\,W(\ST(G,T)).\end{equation}
It follows that $W(\MST(\tilde{G}))\leq 2 W(G^*)=2\,\OPT$.
Now, using the results of~\cite{Alstrup2017} for the special case where $f(x)=tx$ and $t=(2k-1)(1+\eps)$, we can construct a $(2k-1)(1+\eps)$--spanner $\tilde{G}'$ of $\tilde{G}$ which satisfies
\begin{equation}\label{EQ:MSTBound2} W(\tilde{G}') =  O(|T|^\frac1k)\, W(\MST(\tilde{G})).
\end{equation}
Combining these we have the desired estimate
\[ W(\tilde{G}')\leq O(|T|^\frac1k)\;W(G') = O(|T|^\frac1k)\OPT.\]
Finally, $W(G')= O(|T|^\frac1k)  W(ST(G,T))$ follows from \cref{EQ:MSTBound1}.
\end{proof}

Both bounds given in \cref{THM:subsetwise1} are interesting for different reasons.  The first stated bound shows that \cref{ALG:subsetwise} yields an $O(|T|^\frac1k)$--approximation to the optimal solution. 
Dinitz et al.~\cite{dinitz2016label} proved the super-logarithmic hardness for the basic $k$--spanner problem by showing that for every $k \geq 3$ and every constant $\epsilon >0$ it is hard to approximate within a factor better than $2^{\log^{1-\epsilon}n/k}$. That implies that there does not exist any algorithm that can provide a better approximation ratio than \cref{ALG:subsetwise} when $k \geq 3$ unless $NP \subseteq BPTIME(2^{polylog(n)})$.

The second bound gives a better notion of \textit{lightness} of a subsetwise spanner. Typically, lightness of a spanner over all vertices is defined to be $W(G')/W(\MST(G))$, where MST$(G)$ is the minimum spanning tree of $G$.  The notion of lightness for subsetwise spanners suggested by \cref{THM:subsetwise1} is $W(G')/W(\ST(G))$.  When $T=V$,  the minimum spanning tree and the Steiner tree over $V$ are the same, and hence these notions coincide.  However, for $T\subsetneq V$ it is not generally true that $\MST(G)=\ST(G,T)$, but by definition the Steiner tree has smaller weight.  Thus, this notion of lightness for subsetwise spanners stated in terms of the weight of the Steiner tree is more natural.  Klein~\cite{Klein06} uses this notion of lightness for subsetwise spanners of planar graphs, but the results presented here are the first for general graphs. 

The following gives a precise (as opposed to asymptotic) bound for a subsetwise spanner by utilizing the greedy algorithm of Alth\"{o}fer et al.~\cite{Alth90}.

\begin{theorem}\label{THM:subsetwise2}
Let $G(V,E)$, $t>0$, and $T\subset V$ be given.  Using the greedy spanner algorithm as a subroutine, \cref{ALG:subsetwise} yields a subsetwise $(2t+1)$--spanner  $G'(V',E')$ of $G$ in $O(|T|^{3+\frac{1}{t+1}})$ time with total weight
\[W(G') \leq \left(2+\frac{|T|}{t}\right)\OPT.\] Moreover, \[W(G')\leq \left(2+\frac{|T|}{t}\right)\, W(\ST(G,T)).\]
\end{theorem}
\begin{proof}
The greedy algorithm takes $O(mn^{1+\frac{1}{t+1}})$ time, and the rest is similar to the proof of previous theorem.
\end{proof}

Note that in the greedy spanner algorithm of Alth\"{o}fer et~al.~\cite{Alth90}, the minimum spanning tree is a subgraph of the solution. However, for our spanner construction, the minimum Steiner tree might not necessarily be a subset of the final solution. Nevertheless, the produced subsetwise spanner will include a Steiner tree with cost at most twice the optimal one. 

We can try to apply this approach to other multiplicative algorithms available in the literature. For example, using the $O(m)$ time algorithm described in~\cite{baswana2008streaming}, \cref{ALG:subsetwise} yields a subsetwise spanner with running time $O(|T|^2)$, however, the edge size of the output is not given relative to the minimum spanning tree and hence not applicable to our analysis. 

\section{Multi-Level Approximation Algorithms}\label{SEC:MultiLevelCorollaries}

In this section, we illustrate how the subroutines mentioned above can be used in \cref{ALG:composite}, and show several corollaries of the kinds of guarantees one can obtain in this manner. In particular, we give the first weight bounds for multi-level graph spanners.

The case of Steiner trees was discussed at length in \cite{MLST2018} and so is omitted here.

\subsection{Spanners}

If the input graph is planar then we can use the algorithm provided by Klein~\cite{Klein06} to compute a subsetwise spanner for the set of levels we get from the rounding up algorithm. The polynomial time algorithm  in~\cite{Klein06} has constant approximation ratio, assuming that the stretch factor is constant. Hence, we have the following corollary.

\begin{corollary}\label{COR:PlanarWithoutOracleRounding}
Let $G(V,E)$ be a weighted planar graph with $T_{\ell} \subset \dots \subset T_1 \subset V$ given. Let $\eps>0$, and let $A,B$  be as in \cref{EQN:Gamma,EQN:q} assuming uniform edge costs.  Using the planar subsetwise spanner construction of \cite{Klein06} as a subroutine, \cref{ALG:composite} yields an $ABO(\eps^{-4})$--approximation to the MLCTS problem when the constraint requires computing a multiplicative $(1+\eps)$--spanner.  The algorithm runs in $O(\frac{|T_1|\log |T_1|}{\eps})$ time.
\end{corollary}
The proof of this corollary follows from combining the guarantee of Klein with the bound of \cref{THM:RoundingUp}.

Using the approximation algorithm for subsetwise spanners for arbitrary graphs of the previous section, we obtain the following.

\begin{corollary}\label{COR:MLMultSpanner}
Let $G(V,E)$, $k\in\N$, $\eps>0$, and $T_\ell\subset\dots\subset T_1\subset V$ be given, and let $A,B$ be as in \cref{EQN:Gamma,EQN:q} assuming uniform edge costs. Using \cref{ALG:subsetwise} with the spanner construction of \cite{Alstrup2017} as a subroutine, \cref{ALG:composite} yields an $ABO(|T_1|^\frac1k)$--approximation to the MLCTS problem when the constraint requires computing a multiplicative $(2k-1)(1+\eps)$--spanner.  The algorithm runs in $O(|T_1|^{2+\frac1k+\eps})$ time.
\end{corollary}

\begin{proof}
By \cref{THM:subsetwise1}, \cref{ALG:subsetwise} is a $O(|T_j|^\frac1k)$--approximation on any single level, and hence combining this with \cref{THM:RoundingUp}, the conclusion follows.
\end{proof}

For additive spanners, there are algorithms to compute subsetwise spanners of size $O(n|T|^\frac{2}{3})$, $\tilde{O}(n|T|^\frac{4}{7})$ and $O(n|T|^\frac{1}{2})$ for additive stretch 2, 4 and 6, respectively~\cite{Abboud16,Kavitha2017}. Similarly, there is an algorithm to compute a subsetwise $(1+\eps,4)$--spanner of size $O(n\sqrt{\frac{|T|\log n}{\eps}})$~\cite{Kavitha2017}. If we use these algorithms as subroutines in \cref{ALG:composite} to compute subsetwise spanners for different levels, then we have the following corollaries.

\begin{corollary}\label{COR:AdditiveWithoutOracleRounding}
Let $G(V,E)$ be an unweighted graph ($d(e) = 1$ for all $e \in E$) and let $T_{\ell} \subset \dots \subset T_1 \subset V$ be given, and assume uniform edge costs. Then there exist polynomial time algorithms to compute multi-level graph spanners with additive stretch 2, 4 and 6, of size $O(n|T_1|^\frac{2}{3})$, $\tilde{O}(n|T_1|^\frac{4}{7})$, and $O(n|T_1|^\frac{1}{2})$, respectively.
\end{corollary}



\begin{corollary}\label{COR:AlphaBetaWithoutOracleRounding}
Let $G(V,E)$ be an unweighted graph and let $T_{\ell} \subset \dots \subset T_1 \subset V$ and $\eps>0$ be given. Then there exists a polynomial time algorithm to compute multi-level $(1+\eps,4)$--spanners of size $O(n\sqrt{\frac{|T_1|\log n}{\eps}})$.
\end{corollary}

\subsection{$k$--Connected Subgraphs}
Yet another example of constrained terminal sketches are $k$--connected subgraphs~\cite{5438613,laekhanukit2011improved,nutov2012approximating}, in which (similar to the ST problem) a set $T\subset V$ of terminals is given, and the goal is to find the minimum-cost subgraph such that each pair of terminals is connected with at least $k$ vertex-disjoint paths. In this case, the merging operator is $\oplus=\cup$ as in the case of spanners. Nutov~\cite{5438613} presents an approximation algorithm for this problem giving approximation ratio $O(k^2 \log k)$. Laekhanukit~\cite{laekhanukit2011improved} improves the approximation guarantee to $O(k \log k)$  if $|T| \geq k^2$ and shows that the hardest instances of the problem are when $|T| \approx k$. Nutov~\cite{nutov2012approximating} studies the subset $k$--connectivity augmentation problem where given a graph $G$ and a $(k-1)$--connected subgraph $H$, we want to augment some edges to $H$ to make it $k$--connected. The objective is to minimize the size of the set of augmented edges.  If we use the algorithm of~\cite{laekhanukit2011improved} in \cref{ALG:composite} to compute subsetwise $k$--connected subgraphs for different levels, then we have following corollary.

\begin{corollary}\label{COR:kConnectedWithoutOracleRounding}
Let $G(V,E)$ be a graph which is to be sketched on $\ell$ levels with $\ell$ nested terminal sets $T_{\ell} \subset \dots \subset T_1 \subset V$. Then using the algorithm of \cite{laekhanukit2011improved} as a subroutine in \cref{ALG:composite} yields a polynomial time algorithm which computes a multi-level $k$--connected subgraph over the terimals with approximation ratio $O(k \log k)$ provided $|T_1| \geq k^2$.
\end{corollary}

\section{Numerical Experiments}\label{SEC:Experiments}

To evaluate the performance of different variants of 
the composite algorithm we use several experiments. This requires optimal single-level solvers for 
the composite algorithm, which we obtain using an ILP formulation of the problem given in \cite{MLGS_proceeding}. 

\subsection{Setup}
We use the Erd\H{o}s--R\'{e}nyi~\cite{erdos1959random} model to generate random graphs. 
Given a number of vertices, $n$, and probability $p$, the model $\textsc{ER}(n,p)$ assigns an edge to any given pair of vertices with probability $p$. 
An instance of $\textsc{ER}(n,p)$ with $p=(1+\varepsilon)\frac{\ln n}{n}$ is connected with high probability for $\varepsilon>0$~\cite{erdos1959random}).  For our experiments we allow $n$ to range from 5 to 300, and set $\varepsilon = 1$.

For experimentation, we consider only the multiplicative graph spanner version of the MLCTS problem (that is, $\C$ constrains that $G_i$ satisfies $d_{G_i}(u,v) \le td_G(u,v)$ for $u,v \in T_i$) with uniform edge costs based on edge distance ($c(i,e) = i \cdot d(e)$), hence we abbreviate this as MLGS; for similar experimental results on multi-level Steiner trees, see~\cite{MLST2018}. An instance of the MLGS problem is characterized by four parameters: the graph generator, the number of vertices $|V|$, the number of levels $\ell$, and stretch factor $t$ (e.g., ER, $|V| = 80$, $\ell=3$, $t=2$).  As there is randomness involved, we generated 3 instances for every choice of parameters.

We generated MLGS instances with 1 to 6 levels ($\ell \in [6]$), where terminals are selected on each level by randomly sampling $\lfloor {|V| \cdot (\ell-i+1)}/(\ell+1) \rfloor$ vertices on level $i$ so that the size of the terminal sets decreases linearly. As the terminal sets are nested, $T_i$ can be selected by sampling from $T_{i-1}$ (or from $V$ if $i=1$).
We used four different stretch factors in our experiments, $t \in \{1.2, 1.4, 2, 4\}$. Edge weights are randomly selected from $\{1,2,3,\dots,10\}$. 

\subsection{Algorithms and Outputs}

We implemented several variants of \cref{ALG:composite}, which yield different results based on the rounding set $\Q$ as well as the single-level approximation algorithm. In our experiment we used three setups for $\Q$: bottom-up (BU) in which $\Q=\{1\}$, top-down (TD) in which $\Q=\{1,\dots,\ell\}$, and composite (CMP) which selects the optimal set of levels $\Q^*$ as in \Cref{Sec:SingleLevel} (particularly, \cref{ALG:compositeBestQ}).  In the implementation of TD, we add an additional pruning step as we merge the spanners going down the levels; this allows for sparser spanners at each stage.

We used Python 3.5, and utilized the same high-performance computer for all experiments (Lenovo NeXtScale nx360 M5 system with 400 nodes).  When using an oracle for single levels in \cref{ALG:composite}, we use the ILP formulation provided in~\cite{MLGS_proceeding} using CPLEX 12.6.2. 

For each instance of the MLGS problem, we compute the costs of the MLGS returned using the BU, TD, CMP approaches, and also compute the minimum-cost MLGS using the ILP in~\cite{MLGS_proceeding}. For the first set of experiments, we use the ILP as an oracle to find the minimum-cost spanner for each level; in this case, we refer to the results as Oracle BU, TD, and CMP.  In the second set of experiments, we use the metric closure subsetwise spanner \cref{ALG:subsetwise} as the single-level subroutine, which we refer to as Metric Closure BU, TD, and CMP.  We show the performance ratio for each heuristic in the $y$--axis (defined as the heuristic cost divided by OPT), and how the ratio depends on the input parameters (number of vertices $|V|$, number of levels $\ell$, and stretch factors $t$).

Finally, we discuss the running time of the algorithms. All box plots show the minimum, interquartile range and maximum, aggregated over all instances using the parameter being compared.

\subsection{Results}
\label{section:results}

\Cref{oracle_box} shows the impact of different parameters on the results of Oracle TD, BU, and CMP. We see that all variants perform better when the size of the vertex set $|V|$ increases, whereas all variants perform worse as the number of levels $\ell$ increases.  In general, performance decreases as the stretch factor $t$ increases.

\begin{figure}[htp]
    \centering
    \begin{subfigure}[b]{0.28\textwidth}
        \includegraphics[width=\textwidth]{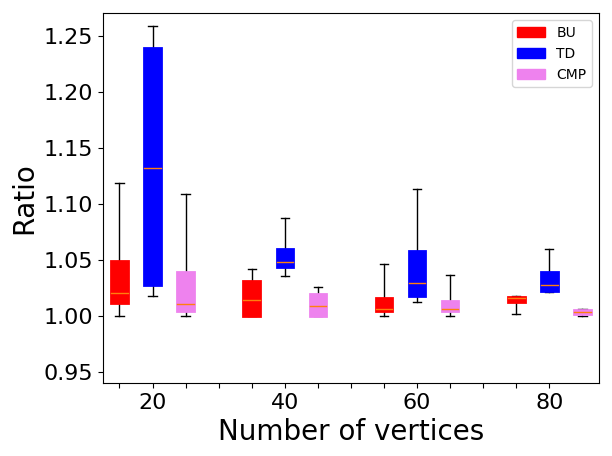}
    \end{subfigure}
    ~
    \begin{subfigure}[b]{0.28\textwidth}
        \includegraphics[width=\textwidth]{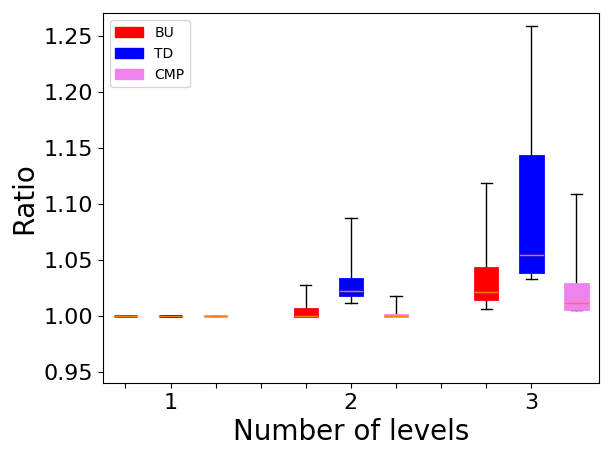}
    \end{subfigure}
    ~
    \begin{subfigure}[b]{0.28\textwidth}
        \includegraphics[width=\textwidth]{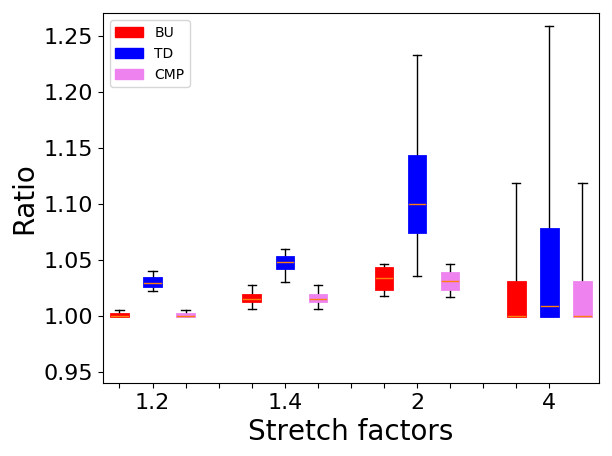}
    \end{subfigure}
    \caption{Performance of Oracle BU, TD and CMP algorithms on Erd\H{o}s--R{\'e}nyi graphs w.r.t.\ the number of vertices, the number of levels, and the stretch factors.} \label{oracle_box}
\end{figure}


The most time-consuming part of the experiment is the execution of the ILP for solving single-level MLGS instances optimally. Hence, we first show the running times of the exact solution of the MLGS instances in \cref{time_box} with respect to the number of vertices $|V|$, number of levels $\ell$, and stretch factors $t$. For all parameters, the running time tends to increase as the size of the parameter increases. 
In particular, the running time with stretch factor 4 (\cref{time_box}, right) was much worse than for lower stretch factors. We can reduce the size of the ILP by removing some constraints based on different techniques discussed in~\cite{MLGS_proceeding}. However, these size reduction techniques are less effective as the stretch factor increases. We show the running times of computing Oracle BU, TD, and CMP solutions in \cref{time_box_heu}. Notice that although the running time of CMP should be worse, sometimes TD takes more time. The reason is that we have an additional edge-pruning step after computing a single-level subsetwise spanner. In TD, every level has this pruning step, which causes additional computation time and affects the runtime adversely when the graph is large.

\begin{figure}[htp]
    \centering
    \begin{subfigure}[b]{0.28\textwidth}
        \includegraphics[width=\textwidth]{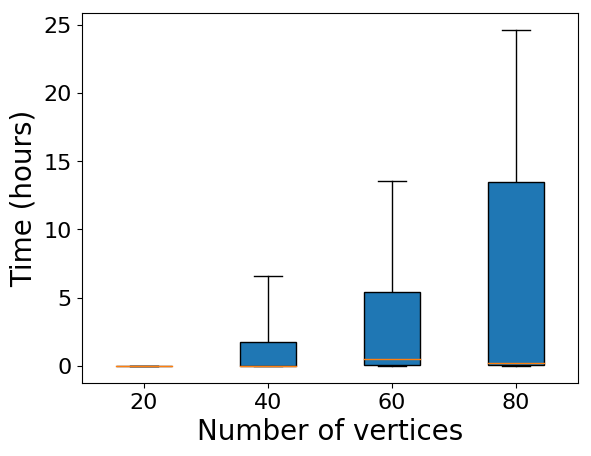}
    \end{subfigure}
    ~
    \begin{subfigure}[b]{0.28\textwidth}
        \includegraphics[width=\textwidth]{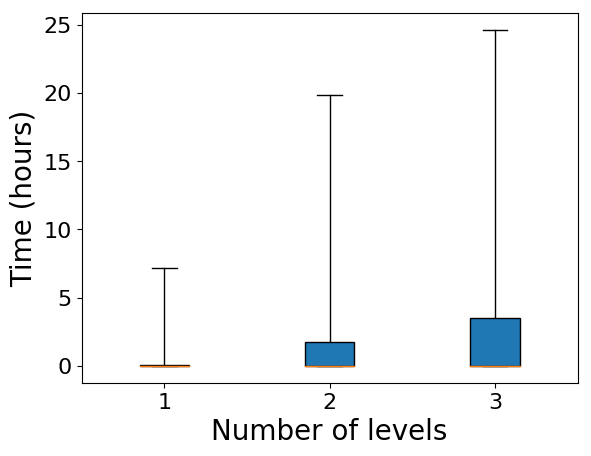}
    \end{subfigure}
    ~
    \begin{subfigure}[b]{0.28\textwidth}
        \includegraphics[width=\textwidth]{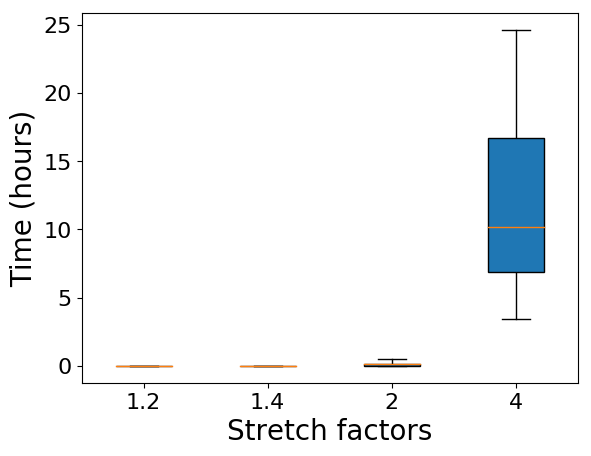}
    \end{subfigure}
    \caption{Experimental running times for computing minimum-cost single-level subsetwise spanners via an ILP w.r.t.\ the number of vertices, the number of levels, and the stretch factors.} \label{time_box}
\end{figure}

\begin{figure}[htp]
    \centering
    \begin{subfigure}[b]{0.28\textwidth}
        \includegraphics[width=\textwidth]{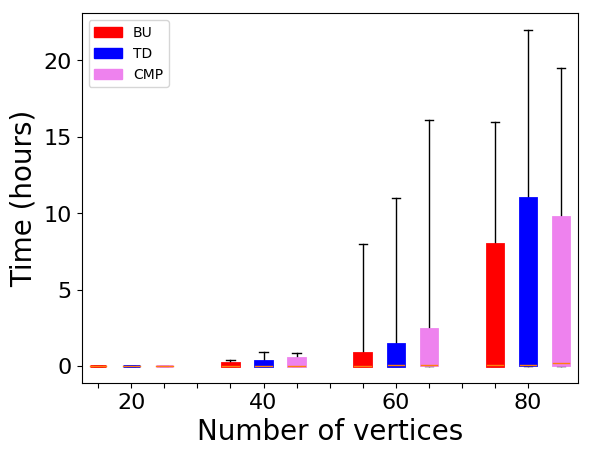}
    \end{subfigure}
    ~
    \begin{subfigure}[b]{0.28\textwidth}
        \includegraphics[width=\textwidth]{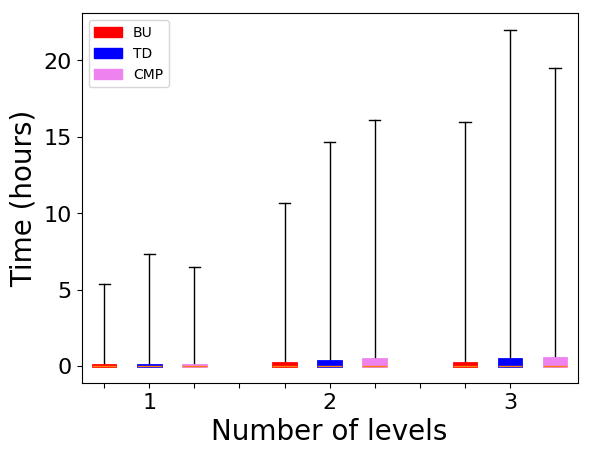}
    \end{subfigure}
    ~
    \begin{subfigure}[b]{0.28\textwidth}
        \includegraphics[width=\textwidth]{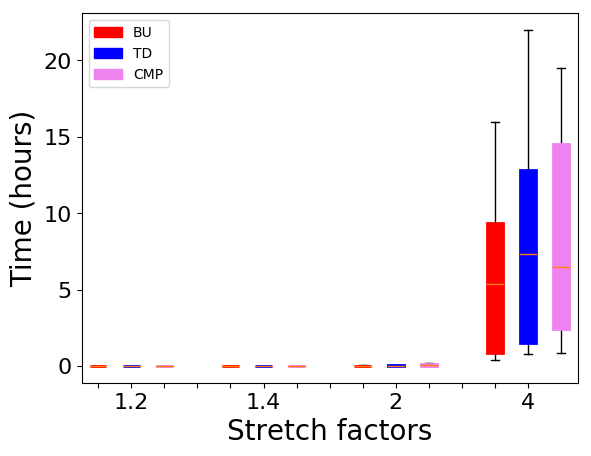}
    \end{subfigure}
    \caption{Experimental running times for computing Oracle BU, TD and CMP solutions w.r.t.\ the number of vertices, the number of levels, and the stretch factors.} \label{time_box_heu}
\end{figure}


The ILP is too computationally expensive for larger input sizes, which is where the approximation algorithm can be particularly useful. We now consider a similar experiment using the metric closure algorithm (\cref{ALG:subsetwise}) to compute subsetwise spanners as described in \Cref{SEC:Spanners}. We show the impact of different parameters in \cref{approx_box}, and one can see that the approximation algorithm performs very well in practice.

\begin{figure}[htp]
    \centering
    \begin{subfigure}[b]{0.28\textwidth}
        \includegraphics[width=\textwidth]{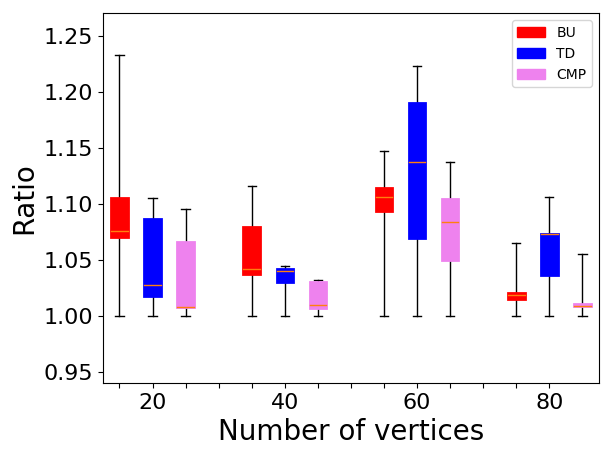}
    \end{subfigure}
    ~
    \begin{subfigure}[b]{0.28\textwidth}
        \includegraphics[width=\textwidth]{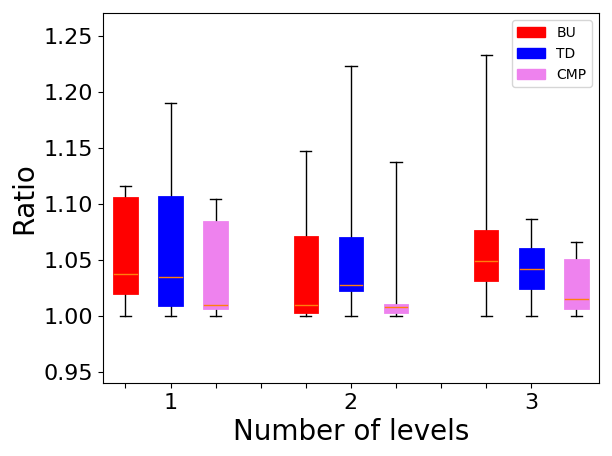}
    \end{subfigure}
    ~
    \begin{subfigure}[b]{0.28\textwidth}
        \includegraphics[width=\textwidth]{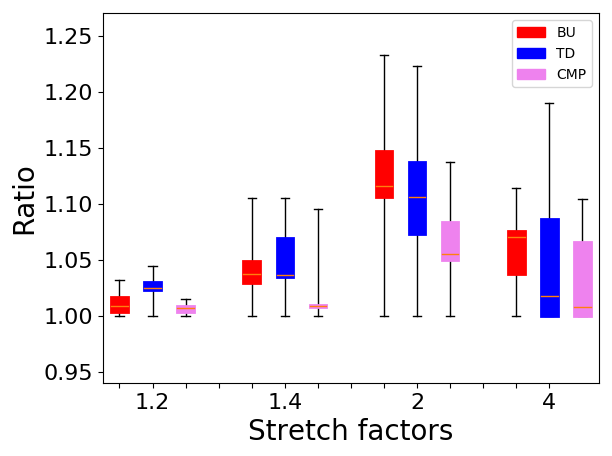}
    \end{subfigure}
    \caption{Performance of the Metric Closure BU, TD and CMP algorithms on Erd\H{o}s--R{\'e}nyi graphs w.r.t.\ the number of vertices, the number of levels, and the stretch factors.} \label{approx_box}
\end{figure}


Our final experiments test the heuristic performance on a set of larger graphs. We generated the graphs using the Erd\H{o}s--R{\'e}nyi model, with $|V| \in \{100, 200, 300\}$, $\ell \in \{2, 4, 6\}$, and with stretch factors $t \in \{1.2, 1.4, 2, 4\}$. Here, the ratio is determined by dividing the BU, TD and CMP cost by $\min\{$BU, TD, CMP$\}$ (as computing the optimal MLGS would be too time consuming). \Cref{consistent_approx_box_ER} shows the performance of the Metric Closure BU, TD and CMP algorithms with respect to $|V|, \ell$ and $t$. \Cref{consistent_time_approx_box_ER} shows the aggregated running times per instance, which significantly worsen as $|V|$ increases. The results indicate that while running times increase with larger input graphs, the number of levels and the stretch factors seem to have little impact on performance. Notably, when the metric closure algorithm is used in place of the ILP for the single-level solver (\cref{consistent_time_approx_box_ER}), the running times decrease for larger stretch factors.   

\begin{figure}[htp]
    \centering
    \begin{subfigure}[b]{0.28\textwidth}
        \includegraphics[width=\textwidth]{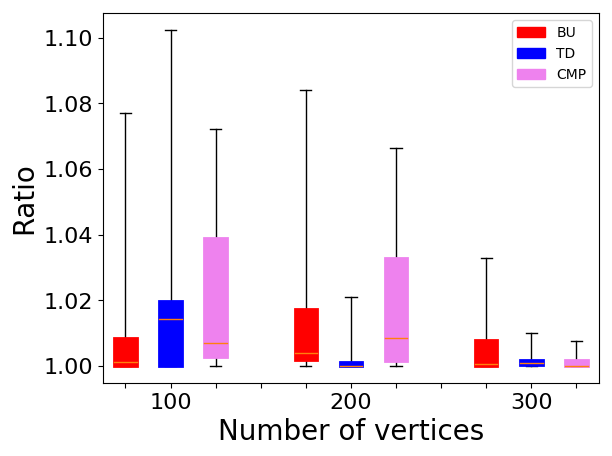}
    \end{subfigure}
    ~
    \begin{subfigure}[b]{0.28\textwidth}
        \includegraphics[width=\textwidth]{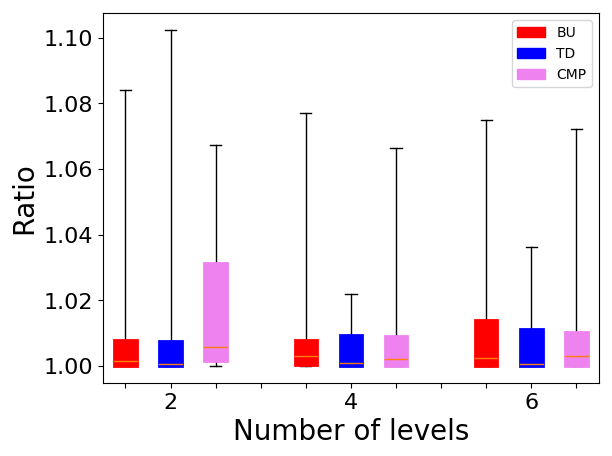}
    \end{subfigure}
    ~
    \begin{subfigure}[b]{0.28\textwidth}
        \includegraphics[width=\textwidth]{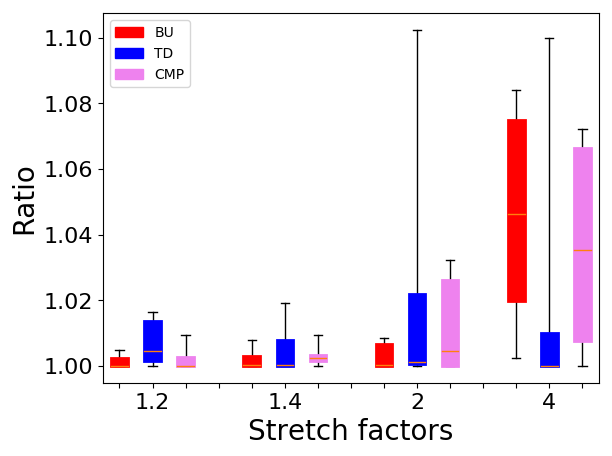}
    \end{subfigure}
    \caption{Performance of Metric Closure bottom-up, top-down and composite on large Erd\H{o}s--R{\'e}nyi graphs w.r.t.\ the number of vertices, the number of levels, and the stretch factors. The ratio is determined by dividing the objective value by min\{BU, TD, CMP\}.} \label{consistent_approx_box_ER}
\end{figure}

\begin{figure}[htp]
    \centering
    \begin{subfigure}[b]{0.28\textwidth}
        \includegraphics[width=\textwidth]{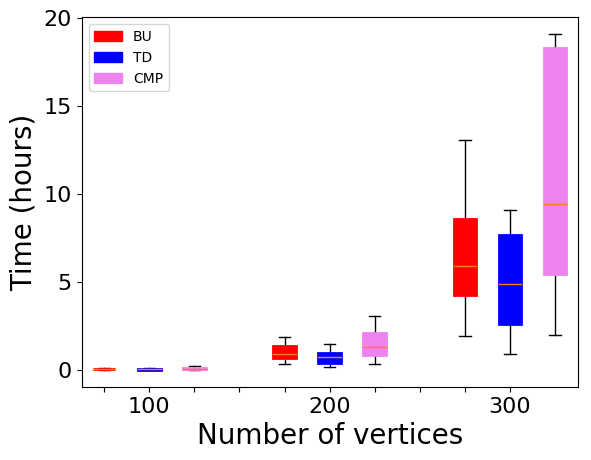}
    \end{subfigure}
    ~
    \begin{subfigure}[b]{0.28\textwidth}
        \includegraphics[width=\textwidth]{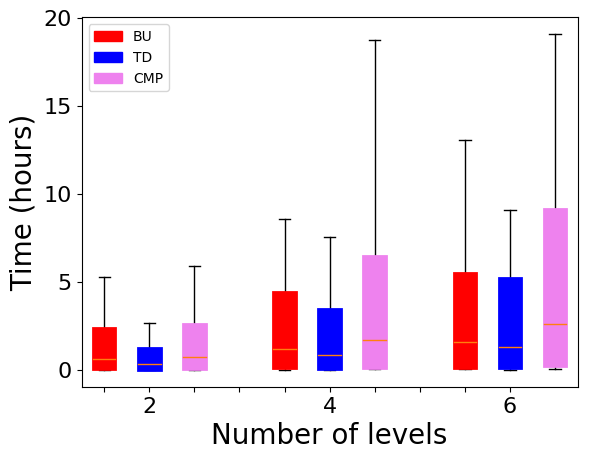}
    \end{subfigure}
    ~
    \begin{subfigure}[b]{0.28\textwidth}
        \includegraphics[width=\textwidth]{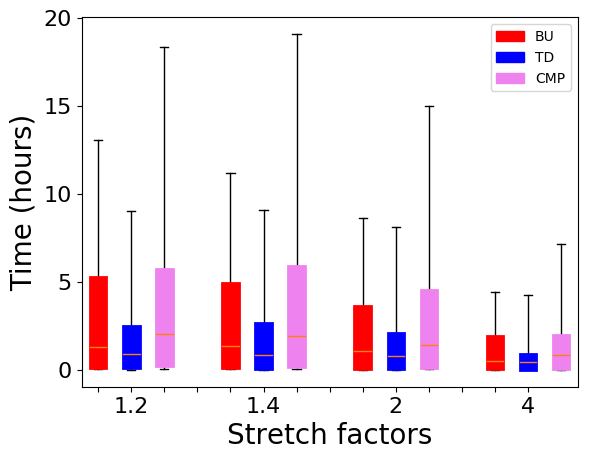}
    \end{subfigure}
    \caption{Experimental running times for computing Metric Closure BU, TD and CMP solutions on large Erd\H{o}s--R{\'e}nyi graphs w.r.t.\ the number of vertices, the number of levels, and the stretch factors.} \label{consistent_time_approx_box_ER}
\end{figure}

\section{Conclusion}\label{SEC:Conclusion}

We have given a general framework for solving multi-level graph sketch problems using a single level solver module. Assuming uniform edge costs, when an oracle is used for the single level module, our algorithm can yield a constant approximation to the optimal multi-level solution that is independent of the number of levels. Using a single-level approximation algorithm as a subroutine, we derive an approximation algorithm for computing multi-level graph spanners which depends on the size of the terminal set but not the number of levels. We also provided the first bounds on the size of subset-wise graph spanners with lightness expressed with respect to the weight of the corresponding Steiner tree. As a result, we showed that the size of multi-level spanners is essentially dominated by the size of the terminal set $T_1$ at the lowest level.

For the future, it would be interesting to look for multi-level algorithms that do not rely on single-level solvers, but which build the solution simultaneously on all levels.


\section*{Acknowledgements}

The research for this paper was partially supported by NSF grants CCF-1740858 and DMS-1839274.

\bibliographystyle{siamplain}
\bibliography{references}


\end{document}